\documentclass[final]{siamltex}
\usepackage{amsmath,amssymb, amsfonts}
\usepackage{graphicx,color}
\usepackage{diagbox}
\usepackage[ruled]{algorithm2e}
\usepackage{algorithmic}
\usepackage{enumitem}
\usepackage{multirow}
\usepackage{subfig}
\usepackage{epstopdf}
\usepackage{comment}
\usepackage{tikz}
\usetikzlibrary{matrix}

\newcommand{\Tr}{\mathrm{Tr}}

\newcommand{\mc}[1]{\mathcal{#1}}
\newcommand{\mf}[1]{\mathfrak{#1}}
\newcommand{\mr}[1]{\mathrm{#1}}

\newcommand{\ud}{\,\mathrm{d}}

\newcommand{\norm}[1]{\lVert#1\rVert}

\newcommand{\RR}{\mathbb{R}}

\global\long\def\ve{\varepsilon}
\global\long\def\R{\mathbb{R}}
\global\long\def\Rn{{\mathbb{R}^{N}}}

\global\long\def\E{\mathbb{E}}

\global\long\def\ra{\rightarrow}
\global\long\def\smooth{C^{\infty}}
\global\long\def\symm{\mathcal{S}^N}
\global\long\def\psd{\mathcal{S}^N_{+}}
\global\long\def\pd{\mathcal{S}^N_{++}}
\global\long\def\dom{\mathrm{dom}\,}
\global\long\def\intdom{\mathrm{int}\,\mathrm{dom}\,}
\global\long\def\Tr{\mathrm{Tr}}

\numberwithin{equation}{section}
\numberwithin{figure}{section}
\newtheorem{thm}{\protect\theoremname}
\newtheorem{lem}[thm]{\protect\lemmaname}
\newtheorem{rem}[thm]{\protect\remarkname}
\newtheorem{prop}[thm]{\protect\propositionname}
\newtheorem{cor}[thm]{\protect\corollaryname}
\newtheorem{assumption}[thm]{Assumption}
\newtheorem{defn}[thm]{Definition}
\newtheorem{notation}[thm]{Notation}
\newtheorem{fact}[thm]{Fact}
\numberwithin{thm}{section}

\providecommand{\corollaryname}{Corollary}
\providecommand{\lemmaname}{Lemma}
\providecommand{\propositionname}{Proposition}
\providecommand{\remarkname}{Remark}
\providecommand{\theoremname}{Theorem}

\allowdisplaybreaks

\title{Bold Feynman diagrams and the Luttinger-Ward formalism via Gibbs
measures \vspace{1.5 mm} \\Part II:  Non-perturbative analysis}
\author{
Lin Lin\thanks{Department of Mathematics, University of California, Berkeley, Berkeley, CA 94720 and Computational Research Division, Lawrence Berkeley National Laboratory, Berkeley, CA 94720. Email: \texttt{linlin@math.berkeley.edu}}
\and Michael Lindsey\thanks{Department of Mathematics, University of California, Berkeley, Berkeley, CA 94720. Email: \texttt{lindsey@math.berkeley.edu}}
}

\begin{document}

\maketitle

\begin{abstract}
Many-body perturbation theory (MBPT) is widely used in quantum physics,
chemistry, and materials science. At the heart of MBPT is the Feynman
diagrammatic expansion, which is, simply speaking, an elegant way of
organizing the combinatorially growing number of terms of a
certain Taylor expansion. In particular, the construction of the `bold Feynman 
diagrammatic expansion' involves the partial resummation to infinite
order of possibly divergent series of diagrams. 
This procedure demands investigation from both the combinatorial
(perturbative) and the analytical (non-perturbative) viewpoints. In Part
II of this two-part series, we approach the analytical 
investigation of the bold diagrammatic expansion in the
simplified setting of Gibbs measures (known as the Euclidean lattice
field theory in the physics literature). 
Using non-perturbative methods, we rigorously construct the Luttinger-Ward formalism 
for the first time, and we prove that the bold diagrammatic 
series can be obtained directly via an asymptotic expansion
of the Luttinger-Ward functional, circumventing the partial
resummation technique. Moreover we prove that the Dyson equation can be derived
as the Euler-Lagrange equation associated with a variational problem
involving the Luttinger-Ward functional. We also establish a number of
key facts about the Luttinger-Ward functional, such as its
transformation rule, its form in the setting of the impurity problem,
and its continuous extension to the boundary of the domain of physical
Green's functions. 
\end{abstract}



\pagestyle{myheadings}
\thispagestyle{plain}

\section{Introduction}\label{sec:intro}

The bold Feynman diagrammatic expansion of 
many-body perturbation theory (MBPT), along with the
many practically used methods in quantum chemistry and condensed matter physics that derive from it, can be formally derived
from the Luttinger-Ward (LW)\footnote{The Luttinger-Ward formalism
is also known as the Kadanoff-Baym
formalism~\cite{BaymKadanoff1961} depending on the context. In this paper we always use
the former.} 
formalism~\cite{LuttingerWard1960}. Since its original proposal in 1960,
the LW formalism has found widespread applicability~\cite{DahlenVanVon2005,Ismail-Beigi2010,BenlagraKimPepin2011,RentropMedenJakobs2016}. 
However, the LW formalism and the LW functional are defined only formally, and this shortcoming poses serious questions
both in theory and in practice. Indeed, the very existence
of the LW functional in the setting of fermionic systems is under debate, with numerical evidence to the
contrary appearing in the past few
years~\cite{KozikFerreroGeorges2015,Elder2014,TarantinoRomanielloBergerEtAl2017,GunnarssonRohringerSchaeferEtAl2017}
in the physics community.

This paper is the second of a two-part series and expands on the work 
in~\cite{LinLindsey2018} that preceded this series.  In Part I, we provided a 
self-contained explanation of MBPT in the setting of the \emph{Gibbs
model} (alternatively known as the `Euclidean lattice field theory' in
the physics literature).  In this setting one is interested in the evaluation of the moments of 
certain Gibbs measures.  While the exact computation
of such possibly high-dimensional integrals is intractable 
in general,
important exceptions are the Gaussian integrals, i.e., integrals for the moments of 
a Gaussian measure, which can be evaluated exactly. 
Perturbing about a 
reference system given by a Gaussian measure, one can 
evaluate quantities of interest 
by a series expansion of Feynman diagrams, which correspond to certain moments of Gaussian measures.
For a specific form of quartic interaction that we refer to as the
\emph{generalized Coulomb interaction}, such a perturbation theory enjoys a
correspondence with the Feynman diagrammatic expansion for the quantum
many-body problem with a two-body
interaction~\cite{NegeleOrland1988,AmitMartin-Mayor2005,AltlandSimons2010}. 
The generalized Coulomb interaction is also of interest in its own right and includes, e.g., the
(lattice) $\varphi^{4}$
interaction~\cite{AmitMartin-Mayor2005,Zinn-Justin2002}, as a special case. 
The combinatorial study of its perturbation theory was the goal of 
Part I.
Nonetheless, the techniques of Part I, and MBPT more broadly, are more
generally applicable to various types of field theories and
interactions.  

The culmination of the developments of Part I is 
the bold diagrammatic expansion, which 
is obtained formally
via a partial resummation technique which sums possibly divergent series
of diagrams to infinite order.
Indeed, the main technical contribution of Part I was to place the 
combinatorial side of this procedure 
on firm footing. One motivation for Part II is to 
interpret the bold diagrams analytically, which we 
accomplish by first constructing the LW
formalism. In fact this construction is non-perturbative 
and valid for rather general forms of interaction. 
Below we focus on the contributions and organization of Part II only.

\subsection{Contributions}
The main contribution of Part II is to establish the LW formalism
rigorously for the first time, in the context of Gibbs measures.
In this setting, the role of the Green's function is
assumed by the two-point correlator.

The construction of the LW functional proceeds via concave duality,
in a spirit similar to that of the Levy-Lieb
construction in density functional theory~\cite{Levy1979,Lieb1983} at
zero temperature and the Mermin functional~\cite{Mermin1965} at finite
temperature, as well as the density matrix functional theory
developed
in~\cite{Baerends2001,SharmaDewhurstLathiotakisEtAl2008,BloechlPruschkePotthoff2013}. 
With careful interpretation, this duality gives 
rise to a one-to-one correspondence between non-interacting and
interacting Green's functions. 
  The LW formalism yields a variational interpretation of the
Dyson equation.
To wit, the free energy can be expressed variationally as a
minimum over all physical Green's functions, and the self-consistent
solution of the Dyson equation yields its unique global minimizer.
We also prove a number of useful properties of the LW functional, such as the
transformation rule, the projection rule, and the continuous extension of the LW functional 
to the boundary of its domain, which can be interpreted as the domain of physical Green's functions. In particular, this last property suggests a novel interpretation of the LW functional as the non-divergent part of the concave dual of the free energy. 
These results allow us to interpret the appropriate analogs of quantum
impurity problems in our simplified setting.  In particular, we prove
that the self-energy is always a sparse matrix for impurity problems,
with nonzero entries appearing only in the block corresponding to the
impurity sites.  Such a result is at the foundation of numerical
approaches such as the dynamical mean field theory (DMFT)~\cite{GeorgesKotliarKrauthEtAl1996,KotliarSavrasovHauleEtAl2006}.

We prove that the bold diagrams for the generalized 
Coulomb interaction can be obtained as asymptotic series
expansions of the LW and self-energy functionals, circumventing the
formal strategy of performing resummation to infinite order.  The proof
of this fact proceeds by proving the existence of such series
non-constructively and then employing the combinatorial results of Part
I to ensure that the terms of these series are in fact given by the bold
diagrams. 

Although the bold diagrammatic expansion (evaluated in terms of the
interacting Green's function, which is always defined) appears to be
applicable in cases where the non-interacting Green's function is
ill-defined, we demonstrate that caution should be exercised in practice
in such cases.  Using a one-dimensional example, we demonstrate that the
approximate Dyson equation obtained via a \emph{truncated} bold
diagrammatic expansion may yield solutions with large error in the
regime of vanishing interaction strength or fail to admit solutions at
all.

\subsection{Outline}
In section \ref{sec:prelim} we review preliminary 
material and definitions needed to understand the results of this paper.

Section \ref{sec:luttingerward} concerns 
the construction of the LW formalism, beginning with 
a discussion of the 
the variational formulation of the free energy and the relevant 
concave duality (section \ref{sec:variation}). This 
is followed by the introduction of the LW functional 
and the Dyson equation (section \ref{sec:lw}). 
Then we introduce several key properties of the LW functional: 
the transformation rule (section \ref{sec:transformation}); the projection rule, accompanied by a discussion of impurity problems (section \ref{sec:projection}); and the 
continuous extension property (section \ref{sec:continuousState}). The proof of the 
continuous extension property, which is the most 
technically demanding part of the paper, is 
postponed to section \ref{sec:boundary}, which 
has its own outline.

Section \ref{sec:bolddiagram} 
concerns the bold diagrammatic expansion.
In section \ref{sec:asymptotic} we prove the 
existence of asymptotic series for the LW 
functional and the self-energy, and in section 
\ref{sec:lwbold} we relate the coefficients of 
the former to the latter. Then for the rigorous 
development of the bold diagrammatic expansion, 
it only remains at this point to prove 
that the asymptotic series for the self-energy 
matches the bold diagrammatic expansion of Part I. 
This is the most involved task of section 
\ref{sec:bolddiagram}. 
In section \ref{sec:p1Review}, we review the 
results that we need from
Part I in a `diagram-free' way 
that should be understandable to the reader 
who has not read Part I, and in 
section \ref{sec:selfenergybold}, we 
establish the claimed correspondence. 
Finally, in section \ref{sec:failbold} 
we illustrate the aforementioned warning about the 
truncation of the bold diagrammatic series 
in cases where the non-interacting 
Green's function is ill-defined.

Relevant background material on convex analysis and the weak convergence of measures is collected in Appendices \ref{sec:convex} and \ref{sec:weakConv}, respectively. The 
proofs of many lemmas are provided in Appendix \ref{sec:appendix}, as noted in the text.

\section{Preliminaries}\label{sec:prelim}

In this section we discuss some preliminary definitions and notations.

\subsection{Notation and quantities of interest}
Throughout we shall let $\mathcal{S}^{N}$, $\mathcal{S}^{N}_+$, and $\mathcal{S}^{N}_{++}$
denote respectively the sets of symmetric, symmetric positive
semidefinite, and symmetric positive definite $N\times N$ real matrices. 
For simplicity we restrict our attention to real matrices, though analogous results can be obtained in the complex Hermitian case.

In this paper we will consider Gibbs measures defined by Hamiltonians $h:\R^N \ra \R \cup \{+\infty\}$ of the form 
\[
h(x) = \frac{1}{2} x^T A x + U(x),
\]
where $A \in \symm$.
The first term represents the quadratic or `non-interacting' part of the Hamiltonian, 
while the second term, $U$, represents the interaction. We define the partition function 
accordingly as  
\begin{equation}
  Z[A, U] = \int_{\RR^{N}} e^{-\frac12 x^{T} A x - U(x)}\ud x.
  \label{eqn:partition}
\end{equation}
For fixed interaction $U$, we may think of the partition function of $A$ alone, i.e., 
as $Z:\symm \ra \R$ sending $A\mapsto Z[A]$. In fact we adopt this perspective 
exclusively for the time being.

The free energy is then defined as a
mapping $\Omega:\mathcal{S}^{N} \rightarrow \R \cup \{-\infty\}$ via
\begin{equation}
\label{eq:OmegaDef}
\Omega[A]:= -\log Z[A] = -\log\int_{\Rn}e^{-\frac{1}{2}x^{T}Ax-U(x)}\,\ud x,
\end{equation}
We denote the \emph{domain} of $\Omega$ by
\[
\dom \Omega := \{A\in\symm\,:\,\Omega[A]>-\infty\},
\]
and the interior of the domain by $\mathrm{int}\,\dom \Omega$. As we will see, 
$\Omega$ is concave in $A$, and this notion of domain is the usual notion from 
convex analysis (see Appendix \ref{sec:convex}), and it is simply the set of 
$A$ such that the integral in Eq.~\eqref{eq:OmegaDef} is convergent.

For $A \in \intdom \Omega$, in fact the integrand in Eq.~\eqref{eq:OmegaDef} 
must decay exponentially, hence we can define
the two-point correlator 
(which we call the \emph{Green's function} by analogy with the 
quantum many-body literature) in terms of $A$ via 
\[
G_{ij}[A] := \frac{1}{Z[A]} \int_{\RR^{N}} x_i x_j \, e^{-\frac12 x^{T} A x
  - U(x)}\ud x,
\]
and the integral on the right-hand side is convergent. More compactly, 
we have a mapping $G : \intdom \Omega \ra \pd$ defined by 
\begin{equation}
  G[A] := \frac{1}{Z[A]} \int_{\RR^{N}} xx^T\, e^{-\frac12 x^{T} A x
  - U(x)}\ud x.
  \label{eqn:green}
\end{equation}
It is important to note that $G[A] \in \pd$ for all $A$.
As we shall see in section \ref{sec:luttingerward}, this constraint defines the domain of `physical' Green's 
functions, in a certain sense. In the discussion below, 
$G$ is also called the interacting Green's function.

In the case of the `non-interacting' Gibbs measure, where $U \equiv 0$,
 all quantities of interest can be computed exactly 
by straightforward multivariate integration. In particular, letting
 $G^{0}[A] := G[A;0]$, 
we have for $A\in \dom\Omega = \pd$ that
\begin{equation}
  G^0 [A] = A^{-1}.
  \label{eqn:gaussian}
\end{equation}
The neatness of this relation is that it motivates the factor of one half included in the 
quadratic part of the Hamiltonian. We refer to $G^0 [A]$ as the \emph{non-interacting 
Green's function}
 associated to $A$, whenever $A \in \pd$. Note that for a general 
interaction $U$, 
$\intdom\Omega$ may contain elements not in $\pd$. For such $A$ there is 
an associated (interacting) Green's function but not a non-interacting Green's function.

In general $G$ can be viewed as the \emph{gradient} of $\Omega$, 
for a suitably defined notion of gradient for functions of symmetric matrices, which 
we now define:
\begin{defn}
For $i,j=1,\ldots,N$, let $E^{(ij)} \in \symm$ be defined by $E^{(ij)}_{kl} = \delta_{ik}\delta_{jl} + \delta_{il}\delta_{jk}$.
For a differentiable function $f:\symm \ra \R$, define the gradient $\nabla f :\symm \ra \symm$ by
\[
\nabla_{ij} f = (\nabla f)_{ij} := \lim_{\delta\to 0}
\frac{f(A+\delta \cdot E^{(ij)})-f(A)}{\delta}.
\]
If $f$ is obtained by restriction from a function $f: \R^{N\times N} \ra \R$, then equivalently $\nabla_{ij}f = \frac{\partial f}{\partial X_{ij}} + \frac{\partial f}{\partial X_{ji}}$.
\label{def:grad1}
\end{defn}

Then on $\dom \Omega$ the gradient map $\nabla \Omega$ is given by 
\begin{equation}
 \nabla_{ij}\Omega[A]=\frac{1}{Z[A]}\int x_i x_j
 \,e^{-\frac{1}{2}x^{T}Ax-U(x)}\,\ud x,
  \label{eqn:gradomega}
\end{equation}
i.e., $G = \nabla \Omega$, as claimed. The notion of gradient of 
Definition \ref{def:grad1} is natural for our setting in that it yields this 
relation. However, it may seem a bit awkward when applied to specific 
computations. Indeed, consider
a function $X\mapsto f(X)$ on $\symm$ that is specified by a formula 
that can be applied to all $N\times N$ matrices and 
in which the roles of $X_{kl}$ and $X_{lk}$ are the same for all $l,k$. 
For instance, such a formula is given by $f(X) = \sum_{ij} X_{ij}^2$. Then 
the usual matrix derivative of $f$, considered as a function on 
$N\times N$ matrices, is given by $\frac{\partial f}{\partial X_{ij}}(X) = 2 X_{ij}$, 
whereas, viewing $f$ as a function on $\symm$ and with notation as specified 
in Definition \ref{def:grad1}, we have $\nabla_{ij} f (X) = 4 X_{ij}$. More generally 
in this situation we have $\nabla_{ij} = 2 \frac{\partial}{\partial X_{ij}}$. Since 
formulas like this arise from the bold diagrammatic expansion (as discussed in 
Part I of this paper), it is convenient then to further define:
\begin{defn}
For a differentiable function $f:\symm \ra \R$, define the matrix derivative $\frac{\partial f}{\partial X} :\symm \ra \symm$ by
\[
\frac{\partial f}{\partial X_{ij}} = \frac12 \nabla_{ij} f.
\]
\label{def:grad2}
\end{defn}
Moreover, this notion of derivative will yield the relation 
\[
\Sigma[G] = \frac{\partial \Phi}{\partial G},
\]
where $\Sigma$ is the self-energy and $\Phi$ is the LW functional, 
as was foreshadowed in Part I.

\subsection{Interaction growth conditions}
\label{sec:growthConditions}
Note that $\dom \Omega$ depends on the shape of $U(x)$.  For example, if $U(x)=0$,
then $\dom \Omega=\mathcal{S}^{N}_{++}$.  If $U(x)=\sum_{i=1}^{N}
x_{i}^4$, then $\dom \Omega=\mathcal{S}^{N}$. Our most basic condition on 
$U$ is the following:
\begin{defn}[Weak growth condition]
\label{def:Ugrowth}
A measurable function $U:\Rn\ra\R$ satisfies the weak growth condition,
if there exists a constant $C_U$ such that $U(x) + C_U( 1 + \Vert x\Vert^2)
\geq 0$ for all $x\in\Rn$, and $\dom \Omega$ is an open set.
\end{defn}

The weak growth condition of Definition~\ref{def:Ugrowth} specifies that
$U$ cannot decay to $-\infty$ faster than quadratically, which ensures in particular that $\dom \Omega$ is
non-empty. The assumption that $\dom \Omega$ is an open set (i.e., 
$\dom \Omega = \mathrm{int}\,\dom \Omega$) will be used later to
ensure that for fixed $U$ there is a one-to-one correspondence between $A$ and $G$ (hence also between non-interacting and interacting Green's functions) over suitable domains.

Note that the condition of 
Definition~\ref{def:Ugrowth} is weaker than the condition
\begin{equation}
  \frac12 x^{T} A x + U(x) \to +\infty, \quad \norm{x}\to +\infty
  \label{eqn:potrequire}
\end{equation}
For instance, if $N=2$ and $U(x)=x_{1}^4$,
then the weak growth condition is satisfied with
$C_{U}=0$, but Eq.~\eqref{eqn:potrequire} is not
satisfied for all $A\in \symm$. In fact, when $U(x)$ only depends on a subset of components of
$x\in \Rn$, we call the Gibbs model an \textit{impurity model} or
\textit{impurity problem}, in analogy with the impurity models of
quantum many-body physics \cite{MartinReiningCeperley2016}, and we call the subset of components
on which $U$ depends the \textit{fragment}. The flexibility of the weak
growth condition will allow us to rigorously establish the LW formalism
for the impurity model. In the setting of the impurity model, the
`projection rule' of Proposition \ref{prop:projection} then allows us to
understand the LW formalism of the impurity model in terms of the
lower-dimensional LW formalism of the fragment and to prove a special
sparsity pattern of the self-energy.

One of our main results (Theorem \ref{thm:contUpTo}) is that the LW functional, which is initially
defined on the set $\pd$ of physical Green's functions, can in fact be
extended continuously to the boundary of $\pd$, a fact which will not be
apparent from the definition of the LW functional. (In fact, this
extension shall be specified by an explicit formula involving
lower-dimensional LW functionals.) But in order for this result to hold,
we need to strengthen the weak growth condition to the following:
\begin{defn}[Strong growth condition]
\label{def:Ugrowth2}
A measurable function $U:\Rn\ra\R$ satisfies the strong growth
condition if, for any $\alpha \in \R$, there exists a constant $b\in \RR$ such that 
$U(x) + b \geq \alpha\Vert x\Vert^2$ for all $x\in \Rn$.
\end{defn}

Note that the strong growth condition ensures that $\dom \Omega =
\mathcal{S}^{N}$ and is hence an open set. If $U$ is a
polynomial function of $x$ and satisfies the strong growth condition,
then Eq.~\eqref{eqn:potrequire} will also be satisfied.

In Section \ref{sec:boundary} we will discuss the precise statement and
proof of the aforementioned continuous extension property. In addition,
a counterexample will be provided in the case where the weak growth
condition holds but the strong growth condition does not. In fact, the
continuous extension property is also valid for impurity models (which
do not satisfy the strong growth condition) via the projection rule
(Proposition \ref{prop:projection}), provided that the interaction
satisfies the strong growth condition when restricted to the fragment.

For the generalized Coulomb interaction considered in Part I, i.e., 
\begin{equation}
  U(x) = \frac{1}{8} \sum_{i,j=1}^{N} v_{ij} x_{i}^2 x_{j}^2,
  \label{eqn:Uterm}
\end{equation}
there is a natural condition on the matrix $v$ that ensures that 
$U$ satisfies the strong growth condition, namely that the matrix $v$ is positive definite. 
We will simply assume that this holds whenever we refer to the generalized Coulomb interaction.
To see that this assumption implies the strong growth condition, first note that $v \succ 0$ 
guarantees in particular that $U$
is a nonnegative polynomial, strictly positive away from $x=0$. Since $U$
is homogeneous quartic, it follows that $U\geq C^{-1}\vert x\vert^{4}$
for some constant C
sufficiently large, which evidently implies the strong growth condition.
Another sufficient assumption is that the entries of $v$ are nonnegative and moreover that the 
diagonal entries are strictly positive. 

Our interest in diagrammatic expansions leads us 
to adopt a further condition on the interaction. Too see why this is necessary, recall 
from Part I that the perturbation about a non-interacting theory ($U
\equiv 0$) involves integrals such as 
\[
\int U(x)\, e^{-\frac12 x^T A x}\, \ud x,
\]
which is clearly undefined if, e.g., $U(x) = e^{x^4}$. In most applications of interest, 
$U(x)$ is only of polynomial growth, but it is sufficient to assume 
growth that is at most exponential 
in the sense of Assumption \ref{assumption:atMostExp}, 
which is actually only needed in section \ref{sec:bolddiagram} for 
our consideration of the bold diagrammatic expansion.
\begin{assumption}[At-most-exponential growth]
\label{assumption:atMostExp}
In this section, we assume that there exist constants $B,C > 0$ such that 
$\vert U(x) \vert \leq B e^{C \Vert x\Vert}$ for all $ x\in \Rn$.
\end{assumption}
Further technical reasons for this assumption will become clear in
section~\ref{sec:bolddiagram}.

\subsection{Measures and entropy: notation and facts}
Let $\mathcal{M}$ be the space of probability measures on $\Rn$ (equipped with 
the Borel $\sigma$-algebra), let $\mathcal{M}_2 \subset \mathcal{M}$ be the subset of probability measures with moments up to second order, and let $\lambda$ denote the Lebesgue measure on $\Rn$. 
For notational convenience we define a mapping that takes the
second-order moments of a probability measure:
\begin{defn}
Define $\mathcal{G}
:\mathcal{M}_2 \ra \mathcal{S}_{+}^{N}$ by $\mathcal{G}(\mu)=\int xx^T\,\ud\mu$.
Writing $\mathcal{G}=(\mathcal{G}_{ij})$, we equivalently have
$\mathcal{G}_{ij}(\mu)=\int x_{i}x_{j}\,\ud\mu$.
\end{defn}

Therefore if $\mu$ is defined via a density
\[
\ud\mu = \rho(x) \ud x, \ \  \mathrm{where}\ \  \rho(x) = \frac{1}{Z[A]}
  e^{-\frac{1}{2}x^{T}Ax-U(x)},
\]
then $\mathcal{G}(\mu) = G[A]$.

We also denote by
\[
\mathrm{Cov}(\mu)=\int x x^T \ud\mu - \left(\int
x\,\ud\mu\right)\left(\int x\,\ud\mu\right)^T
\]
the covariance matrix of $\mu$.

For $\mu \in \mathcal{M}$, let $H$ denote the (differential) entropy
\begin{equation}
\label{eq:entropyDef}
H(\mu)=\begin{cases}
-\int\log\frac{\ud\mu}{\ud\lambda}\,\ud\mu, & \mu\ll \lambda\\
-\infty, & \mathrm{otherwise}
\end{cases}
\end{equation}
where $\frac{\ud\mu}{\ud\lambda}$ denotes the Radon-Nikodym derivative (i.e.,
the probability density function of $\mu$ with respect to the Lebesgue
measure $\lambda$) whenever $\mu\ll \lambda$ (i.e., whenever $\mu$ is absolutely
continuous with respect to the Lebesgue measure). We will often refer to the differential 
entropy as the entropy for convenience.

For $\mu,\nu \in \mathcal{M}$, define the relative entropy $H_\nu(\mu)$ via 
\begin{equation}
\label{eq:relativeEntropyDef}
H_\nu(\mu) = \begin{cases}
-\int\log\frac{\ud\mu}{\ud\nu}\,\ud\mu, & \mu\ll \nu\\
-\infty, & \mathrm{otherwise}.
\end{cases}
\end{equation}
Note carefully the sign convention.\footnote{Our relative entropy is then the negative 
of the Kullback-Leibler divergence, i.e., $H_\nu (\mu) = -D_{\mathrm{KL}} (\mu \Vert \nu)$.} 
The integral in
\eqref{eq:relativeEntropyDef} is well-defined with values in
$\R \cup \{-\infty\}$ for all $\mu,\nu \in \mathcal{M}$.

We now record some useful properties of the relative entropy.

\begin{fact}
\label{fact:relEnt}
For fixed $\nu \in \mathcal{M}$, $H_\nu$ is
non-positive and strictly concave on $\mathcal{M}$, and
$H_\nu(\mu) = 0$ if and only if $\mu = \nu$. Moreover $H_\nu$ is upper semi-continuous 
with respect to the topology of weak convergence; i.e., if the sequence $\mu_k \in \mathcal{M}$ 
converges weakly to $\mu \in \mathcal{M}$, then $\limsup_{k\ra\infty} H_{\nu}(\mu_k) \leq H_{\nu}(\mu)$.
\end{fact}

\begin{proof}
For proofs see \cite{RassoulAgha-Sepp2015}.
\end{proof}

By contrast to the relative entropy, the differential entropy suffers from 
two analytical nuisances.

First, in the definition of the entropy in \eqref{eq:entropyDef}, the entropy may
actually fail to be defined for some measures (which simultaneously
concentrate too much in some area and fail to decay fast enough at
infinity, so the negative and positive parts of the integral are
$-\infty$ and $+\infty$, respectively, and the Lebesgue integral 
is ill-defined). However, Lemma~\ref{lem:entropyMoment} states
that when we restrict to $\mathcal{M}_2$, the integral cannot have an infinite positive
part and is well-defined.
\begin{lem}
\label{lem:entropyMoment}
For $\mu \in \mathcal{M}_2$, if $\mu \ll \lambda$, then the integral in
\eqref{eq:entropyDef} exists (in particular, the positive part of the
integrand has finite integral) and moreover 
\[
H(\mu) \leq \frac{1}{2} \log\left( (2\pi e)^N \det \mathrm{Cov}(\mu) \right) \leq 
\frac{1}{2} \log\left( (2\pi e)^N \det \mathcal{G}(\mu) \right), 
\]
with possibly $H(\mu) = -\infty$. The first inequality is satisfied with
equality if and only if $\mu$ is a Gaussian measure with a positive
definite covariance matrix. The second inequality is satisfied with equality if and only if $\mu$ has mean zero.
\end{lem}

Note that Lemma \ref{lem:entropyMoment} also entails a useful 
bound on the entropy in terms of the second moments, as well as the classical 
fact that Gaussian measures are the measures of maximal entropy subject to second-order 
moment constraints.

The second analytical nuisance of the differential entropy is that we do not 
have the same semi-continuity guarantee as we have for the relative entropy 
in Fact \ref{fact:relEnt}. However, control on second moments allows a
semi-continuity result that will suffice for our purposes.

\begin{lem}
\label{lem:entropyUSC}
Assume that $\mu_j \in \mathcal{M}_2$ weakly converge to $\mu \in
\mathcal{M}$, and that there exists a constant $C$ such that $\mathcal{G}(\mu_j) \preceq C\cdot I_N$ for
all $j$. Then $\limsup_{j\ra \infty} H(\mu_j) \leq H(\mu)$.
\end{lem}
\begin{rem}
In other words, the entropy is upper semi-continuous with respect to the
topology of weak convergence on any subset of probability measures with
uniformly bounded second moments. The subtle difference between the
statements in Fact \ref{fact:relEnt} and Lemma~\ref{lem:entropyUSC} is
due to the fact that the Lebesgue measure $\lambda\notin
\mathcal{M}$.
\end{rem}

The proofs of Lemmas~\ref{lem:entropyMoment} and \ref{lem:entropyUSC} are given in appendix \ref{sec:appendix}.

Finally we record the classical fact that subject to marginal constraints, 
the entropy is maximized by a product measure. In the statement and 
throughout the paper, `$\#$' denotes the pushforward operation on
measures.
\begin{fact} 
\label{fact:productEntropy}
Suppose $p < N$ 
and let $\pi_1:\Rn \ra \R^{p}$ 
and $\pi_2:\Rn \ra \R^{N-p}$ to be
the projections onto the first $p$ and last $N-p$ components,
respectively. Then for $\mu \in \mc{M}_2$, $H(\mu) \leq H(\pi_1 \# \mu) + H(\pi_2 \# \mu)$. 
\end{fact}
\begin{rem}
Note that $\pi_1 \# \mu$ and $\pi_2 \# \mu$ are the marginal distributions of $\mu$
with respect to the product structure $\Rn = \R^p \times \R^{N-p}$.
\end{rem}

See appendix \ref{sec:appendix} for a short proof.

\section{Luttinger-Ward formalism}\label{sec:luttingerward}

This section is organized as follows. In
section~\ref{sec:variation}, we provide a variational expression for the free energy via
the classical Gibbs variational principle. For fixed $U$, this allows us to identify the Legendre dual of $\Omega[A]$, denoted
by $\mathcal{F}[G]$, and to establish a bijection between $A$ and the interacting Green's function $G$. In
section~\ref{sec:lw}, we define the Luttinger-Ward functional and
show that the Dyson equation can be naturally derived by considering the
first-order optimality condition associated to the minimization problem in the variational 
expression for the free energy. Then we prove that the LW 
functional satisfies a number of desirable properties.
First, in section~\ref{sec:transformation} we prove the transformation rule,
which relates a change of the coordinates of the interaction with an appropriate 
transformation of the Green's 
function.  The transformation rule leads to the projection
rule in section~\ref{sec:projection}, which implies the sparsity pattern of the self-energy
for the impurity problem. Up until this point we assume only that $U$ 
satisfy the weak growth condition. Then in section \ref{sec:continuousState} we motivate 
and state 
our result that the LW functional is continuous up to the boundary of $\pd$, for which 
we need the assumption that $U$ satisfies the strong growth condition. The proof 
(as well as a counterexample demonstrating that weak growth is not sufficient) is 
deferred to section \ref{sec:boundary}. Throughout we defer the proofs of some technical 
lemmas to Appendix~\ref{sec:appendix}. Moreover we
will invoke the language of convex analysis following Rockafellar
\cite{Rock} and Rockafellar and Wets \cite{RockWets}. See Appendix 
{\ref{sec:convex}} for 
further background and details.

\subsection{Variational formulation of the free
energy}\label{sec:variation}

The main result in this subsection is given by Theorem~\ref{thm:variation}.

\begin{thm}[Variational structure]\label{thm:variation}
  For $U$ satisfying the weak growth condition, the free energy can
  be expressed variationally via the constrained minimization problem
  \begin{equation}
    \Omega[A]=\inf_{G\in\mathcal{S}^{N}_{+}}\left(\frac{1}{2}\mathrm{Tr}[AG]-\mathcal{F}[G]\right),\label{eq:omegaInf2}
  \end{equation}
  where 
  \begin{equation}
    \mathcal{F}[G]:=\sup_{\mu \in \mathcal{G}^{-1}(G)}\left[H(\mu)-\int
    U\,\ud \mu\right] \label{eq:Fdef}
  \end{equation}
  is the concave conjugate of $\Omega[A]$ with respect to the inner product
  $\langle A,G\rangle = \frac{1}{2}\mathrm{Tr}[AG]$. 
  (Note that by convention 
  $\mathcal{F}[G] = -\infty$ whenever $\mc{G}^{-1}(G)$ is empty, i.e., whenever 
  $G\in\mathcal{S}^{N}\backslash\mathcal{S}_{+}^{N}$.) 
  Moreover $\Omega$ and $\mathcal{F}$ are smooth and strictly concave on their respective domains $\dom\Omega$ and $\pd$. The mapping
  $G[A]:=\nabla\Omega[A]$ is a \textit{bijection} $\dom\Omega \ra
  \mathcal{S}^{N}_{++}$, with inverse given by $A[G]:=\nabla\mathcal{F}[G]$. 
\end{thm}

We first record some technical properties of $\Omega$ in
Lemma~\ref{lem:OmegaConcave}.

\begin{lem}
\label{lem:OmegaConcave}
$\Omega$ is an upper semi-continuous, proper (hence closed) concave function. Moreover,
$\Omega$ is strictly concave and $\smooth$-smooth on $\dom \Omega$.
\end{lem}
\begin{rem}
Recall that a function $f$ on a metric space $X$ is upper semi-continuous if 
for any sequence $x_k \in X$ converging to $x$, we have $\limsup_{k\ra\infty} f(x_k) \leq f(x)$.
\end{rem}

We now turn to exploring the concave (or Legendre-Fenchel) duality associated to $\Omega$.
The following lemma, a version of the classical Gibbs variational
principle~\cite{RassoulAgha-Sepp2015}, is the first step toward identifying the dual of $\Omega$.
\begin{lem}
\label{lem:OmegaDual1}
For any $A\in\mathcal{S}^{N}$, 
\begin{equation}
\Omega[A]=\inf_{\mu\in\mathcal{M}_2}\left[\int\left(\frac{1}{2}x^{T}Ax+U(x)\right)\,\ud\mu(x)-H(\mu)\right].\label{eq:omegaInf1}
\end{equation}
If $A\in \dom \Omega$, the infimum is uniquely attained at $\ud\mu(x)=\frac{1}{Z[A]}e^{-\frac{1}{2}x^{T}Ax-U(x)}\,\ud x$.
\end{lem}
 
\begin{rem}
\label{rem:infSet}
One might wonder whether the infimum in \eqref{eq:omegaInf1} can be
taken over all of $\mathcal{M}$. Note that if $\mu$ does not have a
second moment, it is possible to have both $H(\mu) = +\infty$ and $\int\left(\frac{1}{2}x^{T}Ax+U(x)\right)\,\ud\mu(x) = +\infty$, so the expression in brackets is of the indeterminate form $\infty - \infty$.
The restriction to $\mu \in \mathcal{M}_2$ takes care of this problem because Lemma \ref{lem:entropyMoment} guarantees that $H(\mu) < +\infty$, and by the weak growth condition, the other term in the infimum must be either finite or $+\infty$. Moreover, $\mathcal{M}_2$ is still large enough to contain the minimizer, and restricting our attention to measures with finite second-order moments will be convenient in later developments.
\end{rem}

From the previous lemma we can split up the infimum in
(\ref{eq:omegaInf1}) and obtain
\[
\Omega[A]=\inf_{G\in\mathcal{S}_{+}^{N}}\inf_{\mu\in \mathcal{G}^{-1}(G)}\left[\int\left(\frac{1}{2}x^{T}Ax+U(x)\right)\,\ud\mu (x)
-H(\mu)\right].
\]
Since $\int x^{T}Ax\,\ud\mu=\mathrm{Tr}[\mathcal{G}(\mu)A]$, it follows that
\[
\Omega[A]=\inf_{G\in\mathcal{S}_{+}^{N}}\left(\frac{1}{2}\mathrm{Tr}[AG]+\inf_{\mu\in \mathcal{G}^{-1}(G)}\left[\int U\,\ud\mu-H(\mu)\right]\right).
\]
This proves Eq.~\eqref{eq:omegaInf2} of Theorem~\ref{thm:variation}
using the definition of $\mathcal{F}[G]$ in Eq.~\eqref{eq:Fdef}.

\begin{rem}
\label{rem:Fdomain}
For the perspective of the large deviations theory, we comment
that the construction of $\mathcal{F}$ from the entropy may be
recognizable by analogy to the contraction principle~\cite{RassoulAgha-Sepp2015}. Indeed, the
expression $\int U\,\ud\mu - H(\mu)$ is equal (modulo a constant offset)
to $-H_{\nu_U}(\mu)$, where $\nu_U$ is the measure with density proportional to $e^{-U}$. If one considers i.i.d. sampling from the
probability measure $\nu_U$, by Sanov's theorem $-H_{\nu_U}$ is the
corresponding large deviations rate function for the empirical measure.
The rate function for the second-order moment matrix (i.e.,
$-\mathcal{F}$, modulo constant offset) is obtained via the contraction
principle applied to the mapping $\mu \mapsto \mathcal{G}(\mu)$. This is
analogous to the procedure by which one obtains Cram\'er's theorem from
Sanov's theorem via application of the contraction principle to a map
that maps $\mu$ to its mean~\cite{RassoulAgha-Sepp2015}.  
\end{rem}

Now we record some technical facts about $\mathcal{F}$ in
Lemma~\ref{lem:Ffinite}, which demonstrates in particular that $\mathcal{F}$ 
diverges (at least) logarithmically at the boundary $\partial
\mathcal{S}_{+}^{N} = \mathcal{S}_{+}^{N} \backslash
\mathcal{S}_{++}^{N}$.
\begin{lem}
\label{lem:Ffinite}
$\mathcal{F}$ is finite on $\pd$ and $-\infty$ elsewhere. Moreover, 
\[
\mathcal{F}[G] \leq \frac{1}{2}\log\left[(2\pi e)^N \det G \right] + C_U (1+\mathrm{Tr}\,G)
\]
for all $G \in \pd$.
\end{lem}

Define 
\[
\Psi[\mu]:=H(\mu)-\int U\,\ud\mu,
\]
so $\mathcal{F}[G]=\sup_{\mu\in \mathcal{G}^{-1}(G)}\Psi[\mu]$.
By the concavity of the entropy, $\Psi$ is concave on $\mathcal{M}_2$.
Thus, given $G$, we can in principle solve a
concave maximization problem over $\mu\in\mathcal{M}$ to find
$\mathcal{F}[G]$, with the linear constraint $\mu\in \mathcal{G}^{-1}(G)$. Moreover, this variational representation of $\mathcal{F}$ in terms of the concave function $\Psi$ is 
enough to establish the concavity of $\mathcal{F}$ by abstract considerations. This and other properties of $\mathcal{F}$ are collected in the following.

\begin{lem}
  \label{lem:Fconcave}
$\mathcal{F}$ is an upper semi-continuous, proper (hence closed) concave function on $\mathcal{S}^{N}$.
\end{lem}

Now Eq.~\eqref{eq:omegaInf2} states precisely that $\Omega$ is the \emph{concave
conjugate} of $\mathcal{F}$ with respect to the inner product $\langle
A,G\rangle = \frac{1}{2}\mathrm{Tr}[AG]$, and accordingly we write
$\Omega=\mathcal{F}^{*}$.
Since
$\mathcal{F}$ is concave and closed, we have by Theorem \ref{doubleDual} that 
$\mathcal{F}=\mathcal{F}^{**}=\Omega^{*}$, i.e., $\mathcal{F}$ 
and $\Omega$ are concave duals of one another.
Thus we expect that $\nabla \mathcal{F}$ and $\nabla \Omega$ are
inverses of one another, but to make sense of this claim we need to establish the differentiability of
$\mathcal{F}$. We collect this and other desirable properties of $\mathcal{F}$ in the following:

\begin{lem}
\label{lem:Fdiff}
$\mathcal{F}$ is $C^\infty$-smooth and strictly concave on $\pd$.
\end{lem}

Then Theorem \ref{invSubThm} guarantees that $\nabla\Omega$ is a
bijection from $\dom\Omega \ra \pd$ with its inverse given by $\nabla\mathcal{F}$. This completes the proof of
Theorem~\ref{thm:variation}.

Finally, following Lemma \ref{lem:OmegaDual1}, together with the splitting of
\eqref{eq:omegaInf1} and the $A\leftrightarrow G$ correspondence of
Theorem~\ref{thm:variation}, we observe that the supremum in
\eqref{eq:Fdef} is attained uniquely at the measure $\ud\mu := \frac{1}{Z[A[G]]} e^{-\frac{1}{2}x^{T}A[G]x-U(x)} \ud x$.

\subsection{The Luttinger-Ward functional and the Dyson equation}\label{sec:lw}

According to Lemma \ref{lem:Ffinite}, 
$\mathcal{F}$ should blow up at least logarithmically as $G$ approaches the
boundary of $\mathcal{S}_{++}^{N}$. Remarkably, we can explicitly separate the part that accounts
for the blowup of $\mathcal{F}$ at the boundary. In fact, subtracting
away this part is how we define the Luttinger-Ward (LW) functional for the Gibbs model. We
will see in this subsection that the definition of the Luttinger-Ward functional
can also be motivated by the stipulation that its gradient (the
self-energy) should
satisfy the Dyson equation.

Consider for a moment the case in which $U\equiv0$, so
\[
\mathcal{F}[G]=\sup_{\mu\in \mathcal{G}^{-1}(G)}\left[H(\mu)-\int U\,\ud\mu\right]=\sup_{\mu\in \mathcal{G}^{-1}(G)}H(\mu).
\]
The random variable $X$ achieving the maximum entropy subject to $\E[X_{i}X_{j}]=G_{ij}$
follows a Gaussian distribution, i.e., $X\sim\mathcal{N}(0,G)$. 
It follows that
\[
\mathcal{F}[G]=\frac{1}{2}\log\left((2\pi e)^{N}\det G\right)=\frac{1}{2}\mathrm{Tr}[\log(G)]+\frac{N}{2}\log(2\pi e).
\]

This motivates, for general $U$, the consideration of the
\textit{Luttinger-Ward functional}
\begin{equation}
\Phi[G]:=2\mathcal{\mathcal{F}}[G]-\mathrm{Tr}[\log(G)]-N\log(2\pi e).\label{eq:LWdef}
\end{equation}
For non-interacting systems, $\Phi[G]\equiv 0$ by construction.

Now we turn to establishing the Dyson equation.
Theorem~\ref{thm:variation} shows that for $A\in\dom\Omega$, the
minimizer $G^{*}$ in \eqref{eq:omegaInf2} satisfies 
$A=\nabla\mathcal{F}[G^{*}]=A[G^{*}]$, so the
minimizer is $G^{*}=G[A]$. Recall 
\[
\mathcal{F}[G] = \frac{1}{2}\mathrm{Tr}[\log(G)] + \frac{1}{2}\Phi[G] +\frac{1}{2}N
\log(2\pi e).
\]
Taking gradients and plugging into $A=\nabla\mathcal{F}[G^*]$ yields
\[
0=A-(G^{*})^{-1}- \frac{1}{2} \nabla\Phi[G^{*}].
\]

Define the self-energy $\Sigma$ as a functional of $G$ by
$\Sigma[G]:= \frac{1}{2} \nabla\Phi[G] = \frac{\partial \Phi}{\partial G} [G]$. 
Then we have established that for $G=G[A]$, 
\begin{equation}
  G^{-1}=A-\Sigma[G].
  \label{eqn:dysonLW}
\end{equation}
Moreover, by the strict concavity of $\mathcal{F}$, $G=G[A]$ is the \textit{unique} $G$ solving \eqref{eqn:dysonLW}.

Eq.~\eqref{eqn:dysonLW} is in fact the Dyson equation as in section 3.8 of Part I of this series. 
To see this, recall from Eq.~\eqref{eqn:gaussian} that the non-interacting Green's function $G^{0}$ is
given by $G^{0}=A^{-1}$, so we have 
\[
G^{-1}=(G^{0})^{-1}-\Sigma[G].
\]
 Left- and right-multiplying by $G^{0}$ and $G$, respectively, and
then rearranging, we obtain 
\[
G=G^{0}+G^{0}\Sigma[G]G.
\]
However, Eq.~\eqref{eqn:gaussian} requires $G^{0}$ to be well defined, i.e., $A\in\mathcal{S}_{++}^{N}$.   On the other hand, the Dyson
equation~\eqref{eqn:dysonLW} derived from the LW functional does not
rely on this assumption and makes sense for all $A\in\dom\Omega$. Nonetheless, 
if for fixed $A$ one seeks to approximately solve the Dyson equation for $G$ by inserting 
an ansatz for the self-energy obtained from many-body perturbation theory, one must be 
wary in the case that $A \notin \pd$; see section \ref{sec:failbold}.

\subsection{Transformation rule for the LW
functional}\label{sec:transformation}

Though the dependence of the Luttinger-Ward functional on the
interaction $U$ was only implicit in the previous section, we now explicitly
consider this dependence, including it in our notation as $\Phi[G,U]$.
The same convention will be followed for other functionals without
comment.  Proposition~\ref{prop:LWtransformation} relates a
transformation of the interaction with a corresponding transformation of
the Green's function.

\begin{prop}[Transformation rule]
\label{prop:LWtransformation}Let $G\in\mathcal{S}_{++}^{N}$, 
$U$ be an interaction satisfying the weak growth condition. Let $T$
denote an invertible matrix in $\R^{N\times N}$, as well as the
corresponding linear transformation $\Rn\ra\Rn$. Then 
\[
\Phi[TGT^{*},U]=\Phi[G,U\circ T].
\]
 \end{prop}
\begin{proof}
For $G\in\pd$, note that the supremum in \eqref{eq:Fdef} can be restricted to the set of $\mu \in \mathcal{G}^{-1}(G)$ that have densities with respect to the Lebesgue measure. (Indeed, for any $\mu\in\mathcal{M}_2$ that does not have a density, $H(\mu)-\int U\,\ud\mu = -\infty$.) Then observe
\begin{eqnarray*}
\Phi[G,U] & = & -N\log(2\pi e) - \log\det G + 2 \sup_{\mu\in \mathcal{G}^{-1}(G)}\left[H(\mu)-\int U\,\ud\mu\right] \\
& = & -N\log(2\pi e) - \log\det G - 2 \inf_{\{\rho \,:\, \rho\,\ud x \in \mathcal{G}^{-1}(G)\}}\left[ \int \left(\log \rho + U \right)\,\rho\,\ud x \right] \\
& = &-N\log(2\pi e) - 2 \inf_{\{\rho\,:\,\rho\,\ud x\in \mathcal{G}^{-1}(G)\}}\left[ \int \left(\log \left[(\det G)^{1/2} \rho\right] + U \right)\,\rho\,\ud x \right].
\end{eqnarray*}
Going forward we will denote $C:=-N\log(2\pi e)$. 

Then for $T$ invertible, we have 
\[
\Phi[TGT^*,U] = C - 2 \inf_{\rho\,\ud x\in \mathcal{G}^{-1}(TGT^*)}\left[ \int \left(\log \left[(\det G)^{1/2}\cdot\vert \det T\vert\cdot \rho\right] + U \right)\,\rho\,\ud x \right].
\]
Now observe by changing variables that
\[
\left\{\rho \,:\, \rho\ \ud x \in \mathcal{G}^{-1}(TGT^*)\right\} =
\left\{\vert\det T\vert^{-1}\cdot\rho\circ T^{-1} \,:\, \rho\ \ud x \in \mathcal{G}^{-1}(G)\right\}.
\]
Therefore
\begin{eqnarray*}
\Phi[TGT^*,U] &=& C - 2 \inf_{\rho\,\ud x\in \mathcal{G}^{-1}(G)}\left[ \vert\det T\vert^{-1} \int \left(\log \left[(\det G)^{1/2} \cdot \rho\circ T^{-1}\right] + U \right)\,\rho\circ T^{-1}\,\ud x \right] \\
& = & C - 2 \inf_{\rho\,\ud x\in \mathcal{G}^{-1}(G)}\left[ \int \left(\log \left[(\det G)^{1/2} \cdot \rho \right] + U\circ T \right)\,\rho\,\ud x \right] \\
& = & \Phi[G,U\circ T],
\end{eqnarray*}
as was to be shown.
\end{proof}

\begin{rem}
  Since $T$ is real, the Hermite conjugation $T^{*}$ is the same as the
  matrix transpose, and this is used simply to avoid the notation
  $T^{T}$. 
\end{rem}

From the transformation rule we have the following corollary:
\begin{cor}
\label{cor:fourthOrderScaling}Let $G\in\mathcal{S}_{++}^{N}$, and
consider an interaction $U$ which is a homogeneous polynomial of
degree $4$ satisfying the weak growth condition. For $\lambda>0$, we have 
\[
\Phi[\lambda G,U]=\Phi[G,\lambda^{2}U].
\]
\end{cor}

\subsection{Impurity problems and the projection rule}\label{sec:projection}

For the impurity problem, the interaction only depends on 
a subset of the variables $x_1,\ldots,x_N$, namely the fragment. In such a case,
the Luttinger-Ward functional can be related to a
lower-dimensional Luttinger-Ward functional corresponding to the fragment. This relation, called the
projection rule, is given in Proposition~\ref{prop:projection} below. In the notation, we will now explicitly
indicate the dimension $d$ of the state space associated with the
Luttinger-Ward functional via subscript as in $\Phi_{d}[G,U]$, since
we will be considering functionals for state spaces of different dimensions.
We will follow the same convention for other functionals without comment.

Before we state the projection rule, we record some remarks on the domain of 
$\Omega$ and growth conditions in the context of impurity problems. Suppose that the interaction $U$ depends only on $x_1, \ldots, x_p$, where $p\leq N$, so $U$ can alternatively be considered as a function on $\R ^p$. Notice that even if $U$ satisfies the strong growth condition as a function on $\R^p$, it is of
course \emph{not} true that $\mathrm{dom}\left(
\Omega_N[\,\cdot\,,U] \right) = \symm$. As mentioned above, this provides a natural
reason to consider interactions that do not grow fast in all directions
and motivates the generality of our previous considerations.

In fact, for
\[
A = \left(\begin{array}{cc}
A_{11} & A_{12} \\
A_{12}^T & A_{22}
\end{array}\right),
\]
one can show by Fubini's theorem, integrating out the last $N-p$ variables in \eqref{eq:OmegaDef}, that $A\in \mathrm{dom}\left( \Omega_N[\,\cdot\,,U] \right)$ if and only if both 
\[
A_{22} \in \mathcal{S}^{N-p}_{++}\ \mathrm{and}\  A_{11} - A_{12} A_{22}^{-1} A_{12}^T \in \mathrm{dom}\,\left( \Omega_p[\,\cdot\,,U] \right).
\]
Moreover, one can show that for such $A$,
\[
\Omega_N[A,U] = \Omega_p \left[A_{11} - A_{12} A_{22}^{-1} A_{12}^T,\ U
\right] + \frac{1}{2}\log( (2\pi)^{p-N} \det A_{22}).
\]
Therefore, if $\mathrm{dom}\,\left( \Omega_p[\,\cdot\,,U(\,\cdot\,,0)] \right)$ is open, then so is $\mathrm{dom}\left( \Omega_N[\,\cdot\,,U] \right)$. It follows that if $U$ satisfies the weak growth condition as a function on $R^p$, then $U$ also satisfies the weak growth condition as a function on $\R ^N$.

\begin{prop}[Projection rule]
\label{prop:projection}
Let $p\leq N$. Suppose that $U$ depends only
on $x_1,\ldots,x_p$ and satisfies the weak growth condition. Hence we can think of $U$ as a function on both $\Rn$ and $\R^p$.
Then for $G\in\pd$,
\[
\Phi_{N}\left[G,U\right]=\Phi_{p}\left[G_{11},U\right],
\]
where $G_{11}$ is the upper-left $p\times p$ block of $G$.
\end{prop}

\begin{rem}
If $U$ can be made to depend only on $p\leq N$ variables by linearly
changing variables, then we can use the projection rule in combination
with the transformation rule (Proposition \ref{prop:LWtransformation})
to reveal the relationship with a lower-dimensional Luttinger-Ward
functional, though we do not make this explicit here with a formula.
\end{rem}

\begin{cor}
\label{cor:projection}
Let $p\leq N$, and $P$ be the orthogonal projection onto the subspace $\mathrm{span}\,\{e_1^{(N)},\ldots,e_p^{(N)}\}$. Suppose that $U(\,\cdot\,,0)$ satisfies the weak growth condition. Then for $G \in \pd$,
\[
\Phi_{N}\left[G,U\circ P\right]=\Phi_{p}\left[G_{11},U(\,\cdot\,,0)\right],
\]
where $G_{11}$ is the upper-left $p\times p$ block of $G$.
\end{cor}

\begin{proof}
(Of Proposition \ref{prop:projection}.)
First we observe that we can assume that $G$ is block-diagonal. To see
this, let $G\in \pd$, and write
\[
G = \left(\begin{array}{cc}
G_{11} & G_{12} \\
G_{12}^T & G_{22}
\end{array}\right).
\]
Then block Gaussian elimination reveals that
\[
G = 
\left(\begin{array}{cc}
I & 0\\
G_{12}^{T}G_{11}^{-1} & I
\end{array}\right)\left(\begin{array}{cc}
G_{11} & 0\\
0 & G_{22}-G_{12}^{T}G_{11}^{-1}G_{12}
\end{array}\right)\left(\begin{array}{cc}
I & G_{11}^{-1}G_{12}\\
0 & I
\end{array}\right).
\]
Define 
\[
T:= \left(\begin{array}{cc}
I & 0\\
G_{12}^{T}G_{11}^{-1} & I
\end{array}\right),
\ \ 
\widetilde{G} :=
\left(\begin{array}{cc}
G_{11} & 0\\
0 & G_{22}-G_{12}^{T}G_{11}^{-1}G_{12}
\end{array}\right),
\]
so $G = T\widetilde{G}T^*$. Then by the transformation rule, we have
\[
\Phi_N[G,U] = \Phi_N[\widetilde{G},U\circ T] = \Phi_N[\widetilde{G},U],
\]
where the last equality uses the fact that $U$ depends only on the first $p$ arguments, which are unchanged by the transformation $T$.

Since $\widetilde{G}$ is block-diagonal with the same upper-left block as $G$, we have reduced to the block-diagonal case, as claimed, so now assume that $G\in \pd$ with
\[
G = \left(\begin{array}{cc}
G_{11} & 0 \\
0 & G_{22}
\end{array}\right).
\]

Recall the expression for $\mathcal{F}_N$:
\[
\mathcal{F}_N [G,U] = \sup_{\mu\in \mathcal{G}_N^{-1}(G)}\left[H(\mu)-\int U\,\ud\mu\right].
\]

Next define $\pi_1:\Rn \ra \R^{p}$ and $\pi_2:\Rn \ra \R^{N-p}$ to be
the projections onto the first $p$ and last $N-p$ components,
respectively. Then with `$\#$' denoting the pushforward operation on
measures, $\pi_1 \# \mu$ and $\pi_2 \# \mu$ are the marginals of $\mu$
with respect to the product structure $\Rn = \R^p \times \R^{N-p}$.
Now recall Fact \ref{fact:productEntropy}, in particular the inequality $H(\mu) \leq H(\pi_1 \# \mu) + H(\pi_2 \# \mu)$. Also note that if $\mu
\in \mathcal{G}_N^{-1}(G)$, then $\pi_1 \# \mu \in \mathcal{G}_p^{-1}(G_{11})$ and $\pi_2 \#
\mu \in \mathcal{G}_{N-p}^{-1}(G_{22})$. Finally observe that since $U$ depends
only on the first $p$ arguments, $\int U\,\ud\mu = \int U\,d(\pi_1 \#
\mu)$ for any $\mu$. Therefore
\begin{eqnarray*}
\mathcal{F}_N [G,U] & \leq & \sup_{\mu\in \mathcal{G}_N^{-1}(G)}\left[H(\pi_1 \# \mu) + H(\pi_2 \# \mu) -\int U\,d(\pi_1 \# \mu) \right] \\
& \leq & \sup_{\mu_1 \in \mathcal{G}_p^{-1}(G_{11})} \left[H(\mu_1) -\int U\,\ud\mu_1 \right] + \sup_{\mu_2 \in \mathcal{G}_{N-p}^{-1}(G_{22})}\left[H(\mu_2)  \right] \\
& = & \mathcal{F}_p [G_{11},U] + \frac{1}{2}\log( (2\pi e)^{N-p} \det G_{22}).
\end{eqnarray*}
Since $\det G = \det G_{11} \det G_{22}$, it follows that 
\[
\Phi_N [G,U] \leq \Phi_p [G_{11},U].
\]

For the reverse inequality, let $\mu_1$ be arbitrary in $\mathcal{G}_p^{-1}(G_{11})$, and consider $\mu:= \mu_1 \times \mu_2$, where $\mu_2$ is given by the normal distribution with mean zero and covariance $G_{22}$. Then 
\[
\mathcal{F}_N [G,U] \geq H(\mu) - \int U\,\ud\mu = H(\mu_1)- \int U\,\ud\mu_1 + \frac{1}{2}\log( (2\pi e)^{N-p} \det G_{22}).
\]
Since $\mu_1$ is arbitrary in $\mathcal{G}_p^{-1}(G_{11})$, it follows by taking the supremum over $\mu_1$ that 
\[
\mathcal{F}_N [G,U] \geq \mathcal{F}_p [G_{11},U] + \frac{1}{2}\log( (2\pi e)^{N-p} \det G_{22}),
\]
which implies 
\[
\Phi_N [G,U] \geq \Phi_p [G_{11},U].
\]
\end{proof}

\begin{rem}
The proof suggests that for $U$ depending only on the first $p$
arguments and $G$ block-diagonal, the supremum in the definition of $\mathcal{F}$ is attained by a product measure, which is perhaps not surprising. The proof also suggests, however, that for such $U$ and general $G$, the supremum is attained by taking a product measure and then `correlating' it via the transformation $T$.
\end{rem}

For the impurity problem, Proposition
\ref{prop:projection} immediately implies that the self-energy has a particular 
sparsity pattern:
\begin{cor}
\label{cor:SEsparse}
Let $p\leq N$ and suppose that $U$ (satisfying the weak growth condition) depends only on $x_1,\ldots,x_p$. Then 
\[
\Sigma_N [G,U] = \left(\begin{array}{cc}
\Sigma_p [G_{11},U] & 0 \\
0 & 0
\end{array}\right).
\]
\end{cor}

For example, consider $U(x) = \frac{1}{8} \sum_{ijkl} v_{ij} x_i^2  x_j^2$.
Here the stipulation that $U$ depend only on the first $p$ arguments
corresponds to the stipulation that $v_{ij} = 0$ unless $i,j \leq
p$. For such an interaction, in the bold diagrammatic expansion for $\Phi$
and $\Sigma$, any term in which $G_{ij}$ appears will be zero unless
$i,j\leq p$. This is a non-rigorous perturbative explanation of the fact that $\Phi$
depends only on the upper-left block of $G$, which in turn explains the
sparsity structure of $\Sigma$, as well as the fact that $\Sigma$ also
depends only on the upper-left block of $G$. However, the developments of this 
section apply to interactions $U$ of far greater generality and which may indeed be
non-polynomial, hence not admitting of a bold diagrammatic expansion.

\subsection{Continuous extension of the LW functional to the boundary}
\label{sec:continuousState}
The 
discussion in this subsection is only heuristic, and the proofs of the 
theorems stated here are deferred to section \ref{sec:boundary}.

Now in section~\ref{sec:variation} we saw that 
the functional $\mathcal{F}[G]$ diverges at the boundary
$\partial \mathcal{S}_{+}^{N} = \mathcal{S}_{+}^{N} \backslash
\mathcal{S}_{++}^{N}$. On the other hand, the projection rule together with the
transformation rule, motivates the formula by which we can extend $\Phi$
continuously up to the boundary $\partial \mathcal{S}_{+}^{N}$.

Indeed, suppose that $T^{(j)} \ra P$, where $T^{(j)}$ is invertible and $P$
is the orthogonal projection onto the first $p$ components, as in Corollary \ref{cor:projection}. Then for $G\in \pd$, 
\[
 \Phi_N [T^{(j)}G(T^{(j)})^*,U] = \Phi_N [G,U\circ T^{(j)}] .
\]
By naively taking limits of both sides, we expect that
\[
\Phi_N [PGP, U] = \Phi_N [G,U\circ P]
\]
where $G_{11}$ is the upper-left $p\times p$ block of $G$. Then by the projection rule we expect
\[
 \Phi_N \left[ \left(\begin{array}{cc} G_{11} & 0 \\
0 & 0
\end{array}\right), U \right]
= \Phi_p [G_{11},U(\,\cdot\,,0)],
\]
where $G_{11}$ is the upper-left $p\times p$ block of $G$. After possibly changing coordinates via 
the transformation rule, this formula provides a general recipe for evaluating the LW functional on 
the boundary $\partial \psd$, which is the content of Theorem \ref{thm:contUpTo} below.

Unfortunately, there are nontrivial analytic difficulties that are hidden by this heuristic derivation.
In fact there exists an interaction $U$ satisfying the weak growth 
condition for which the continuous extension property fails. Since the discussion of this counterexample 
is somewhat involved, it is postponed to section \ref{subsec:counterexample}. However, 
the continuous extension property is true for $U$ satisfying the strong growth condition of Definition \ref{def:Ugrowth2}.

Before stating the continuous extension property in Theorem \ref{thm:contUpTo}, we provide a 
more careful discussion of the structure of the boundary $\partial \psd$. Consider a $q$-dimensional 
subspace $K$ of $\Rn$, and let $p=N-q$.
Then the set
\[
S_{K}:=\left\{ G\in\mathcal{S}_{+}^{N}\,:\,\ker G=K\right\} 
\]
forms a `stratum' of the boundary of $\mathcal{S}_{+}$, which is
itself isomorphic to the set of $p\times p$ positive definite matrices.
In turn, one can consider boundary strata (of smaller dimension) nested
inside of $S_{K}$.

We will show that the restriction of the Luttinger-Ward function to
such a stratum is precisely the Luttinger-Ward function for a lower-dimensional
system. To this end, fix a subspace $K$ and choose any orthonormal
basis $v_{1},\ldots,v_{p}$ for $K^{\perp}$. (The choice of basis
is not canonical but can be made for the purpose of writing down results explicitly.)
Define $V_{p}:=[v_{1},\ldots,v_{p}]$. We use this notation to indicate both the matrix 
and the corresponding linear map.

\begin{thm}[Continuous extension, I]
  \label{thm:contUpTo}Suppose that $U$ is continuous and satisfies the strong growth condition. With notation as in the preceding discussion,
$\Phi_{N}[\,\cdot\,,U]$ extends continuously to $S_{K}$ via the rule
\[
\Phi_{N}\left[G,U\right]=\Phi_{p}\left[V_{p}^{*}GV_{p},U\circ V_{p} \right]
\]
 for $G\in S_{K}$. Consequently, $\Phi_{N}[\,\cdot\,,U]$ extends continuously
to all of $\mathcal{S}_{+}^{N}$.\end{thm}
\begin{rem}
\label{rem:growthSC}
We interpret the extension rule as to set $\Phi_N [0,U] = \Phi_0[U] :=
-2 \cdot U(0)$. Moreover, it will become clear in the proof that even for continuous interactions $U$ that do not satisfy the strong growth condition, the extension is still lower semi-continuous on $\psd$ and continuous on $\pd \cup \{0\}$.
\end{rem}

Changing coordinates via Proposition \ref{prop:LWtransformation}, we see that Theorem \ref{thm:contUpTo} is actually equivalent to the following:

\begin{thm}[Continuous extension, II]
\label{thm:contUpToSpecialCase}Suppose that $U$ is continuous and satisfies the strong growth condition. For $G\in\mathcal{S}_{++}^{p}$, $\Phi[\,\cdot\,,U]$ extends continuously via the rule
\[
\Phi_{N}\left[\left(\begin{array}{cc}
G & 0\\
0 & 0
\end{array}\right),U\right]=\Phi_{p}\left[G,U(\cdot,0)\right].
\]
\end{thm}

Once again we comment that proof is deferred to section \ref{sec:boundary}.

\section{Bold diagram expansion for the generalized Coulomb interaction}\label{sec:bolddiagram}

Using the Luttinger-Ward formalism, in this section we prove that the
bold diagrammatic expansions from Part I of the self-energy and the LW functional 
(for the generalized Coulomb interaction \eqref{eqn:thmSeries}) can indeed 
be interpreted as asymptotic series expansions in the interaction 
strength at fixed $G$.
This provides a rigorous interpretation of the bold expansions that is not merely 
combinatorial.  Recall that when each $G$ in the bold diagrammatic expansion 
of the self-energy is
further expanded using $G^{0}$ and $U$, the resulting expansion should
be formally the same as the bare diagrammatic expansion of the self energy. The 
combinatorial argument in section 4 of Part I guaranteeing this fact 
does not need to be repeated in this setting, and we will be able to directly use 
Theorem 4.12 from Part I. The remaining hurdles are analytical, not 
combinatorial.

We summarize the results of this section as follows.
 \begin{thm}
 \label{thm:boldDiag}
 For any continuous interaction $U:\Rn \ra \R$ satisfying the weak growth condition and 
 any $G \in \pd$, the LW functional and the self-energy have asymptotic series expansions as 
 \begin{equation}
 \label{eqn:thmSeries}
 \Phi[G,\ve U] = \sum_{k=1}^\infty \Phi^{(k)}[G,U]\ve^k,\quad \Sigma[G,\ve U] = \sum_{k=1}^\infty \Sigma^{(k)}[G,U]\ve^k.
 \end{equation}
 Moreover, for $U$ a homogeneous quartic polynomial, the coefficients of the 
 asymptotic series satisfy
 \begin{equation}
 \label{eq:boldRelation}
 \Phi^{(k)}[G,U] = \frac{1}{2k} \Tr \left[ G \Sigma^{(k)}[G,U] \right].
 \end{equation}
 If $U$ is moreover a generalized Coulomb interaction~\eqref{eqn:Uterm}, 
 we have (borrowing the language of Part I) that
 \begin{equation}
 \label{eqn:thmSeries2}
 \Sigma^{(k)}_{ij}[G,U] = \sum_{\Gamma_{\mr{s}} \in \mf{F}_2^{\mr{2PI}},\,\mathrm{order}\,k} \frac{\mathbf{F}_{\Gamma_{\mr{s}}}(i,j)}{S_{\Gamma_{\mr{s}}}},
 \end{equation}
 i.e., $\Sigma^{(k)}$ is given the sum over bold skeleton diagrams of order $k$ with bold 
 propagator $G$ and interaction $v_{ij} \delta_{ik} \delta_{jl}$.
 \end{thm}
 \begin{rem}
 For a series as in Eq.~\eqref{eqn:thmSeries} to be asymptotic means that 
 the error of the $M$-th partial sum is $O(\ve^{M+1})$ as $\ve \ra 0$.
 \end{rem}

Since $U$ is fixed, for simplicity in the ensuing discussion  we will omit the
dependence on $U$ from the notation
via the definitions $\Phi_G(\ve) :=
\Phi [G,\ve U]$, $\Sigma_G(\ve) = \Sigma[G,\ve U]$, and $A_G(\ve) :=
A[G,\ve U]$. We will also denote the series coefficients via
$\Phi^{(k)}_G := \Phi^{(k)}[G,U]$ and $\Sigma^{(k)}_G :=
\Sigma^{(k)}[G,U]$. In this notation, our asymptotic series take the
form
\begin{equation}
\label{eq:LWSEbold}
\Phi_G (\ve) = \sum_{k=1}^\infty \Phi^{(k)}_G \ve ^k, \quad \Sigma_G (\ve) = \sum_{k=1}^\infty \Sigma^{(k)}_G \ve ^k.
\end{equation}

\begin{notation}
Note carefully that in this section the superscript ${(k)}$ is merely a notation and \emph{does not} indicate the $k$-th derivative. Such derivatives will be written 
out as $\frac{\ud^k}{\ud \ve^k}$.
\end{notation}

Now we outline the remainder of this section. In section~\ref{sec:asymptotic} 
we prove that the LW functional and the self-energy do indeed 
admit asymptotic series expansions. 
In section \ref{sec:lwbold} we prove the relation between the 
LW and self-energy expansions for quartic interactions, namely Eq.~\eqref{eq:boldRelation}.
Interestingly, this relation---which is well-known formally based on
diagrammatic observations---was originally assumed to be 
true to obtain a formal derivation of the LW
functional~\cite{LuttingerWard1960,MartinReiningCeperley2016}.  Our
proof here does not rely on any diagrammatic manipulation, only
making use of the transformation rule and the quartic nature of the
interaction $U$. Similar relations for homogeneous polynomial interactions of different
order could easily be obtained. 
Next, in section~\ref{sec:p1Review}, we summarize and expand on the necessary results 
from Part I 
in diagram-free language; this both reduces the prerequisite knowledge needed for 
the remainder of the section and clarifies the arguments that follow.
Finally, in section~\ref{sec:selfenergybold} 
we prove that when $U$ is a generalized Coulomb interaction, 
the series for the self-energy is in fact the bold diagrammatic
expansion of section 4 of Part I.

\subsection{Existence of asymptotic series}\label{sec:asymptotic}

In this section we assume that $U$ is continuous and satisfies the weak 
growth condition.
We first prove the following pair of lemmas.

\begin{lem}
\label{lem:Aeps}
For any $G\in\pd$, $A_G (\ve) \ra G^{-1}$ as $\ve \ra 0^+$.
\end{lem}

\begin{lem}
\label{lem:epsCts}
For $G\in\pd$, all derivatives of the functions $\Phi_G:(0,\infty) \ra \R$ 
and $\Sigma_G:(0,\infty) \ra \R^{N\times N}$ extend continuously to $[0,\infty)$.
\end{lem}

We will convey the continuous extension of the derivatives of $\Phi_G$ to the origin by the notation $\Phi^{(k)}_G := \Phi^{(k)}_G (0)$, and similarly for the self-energy $\Sigma_G^{(k)} := \Sigma_G^{(k)}(0)$.  From the preceding it will follow that the series \eqref{eq:LWSEbold} are indeed asymptotic series in the 
following sense:

\begin{prop}
\label{prop:asymptoticSeries}
For any nonnegative integer $M$, $\Phi_G(\ve) - \sum_{k=1}^M \Phi^{(k)}_G \ve^k = O(\ve^{M+1})$ and 
$\Sigma_G(\ve) - \sum_{k=1}^M \Sigma^{(k)}_G \ve^k = O(\ve^{M+1})$ as $\ve \ra 0^+$.
\end{prop}
\begin{proof}
Consider any function $f:[0,\infty)\ra\R$ with all derivatives extending
continuously up to the boundary (and so defined at 0). Let $\delta>0$,
so for $\ve \in (\delta,1]$ we know  by the Lagrange error bound that
\[
\left\vert f(\ve) - \sum_{k=0}^M f^{(k)}(\delta) (\ve -\delta)^k
\right\vert \leq C (\ve - \delta)^{M+1} \leq C \ve^{M+1},
\]
where $C$ is a constant that depends only on a uniform bound on
$\left(\frac{\ud}{\ud \ve}\right)^{k+1} f$ over $[0,1]$ (the existence of which is guaranteed by the
continuous extension property). Simply taking the limit of our
inequality as $\delta \ra 0^+$, and again employing the continuous
extension property, yields that $\left\vert f(\ve)-\sum_{k=0}^M
f^{(k)}(0) \ve^k \right\vert \leq C\ve^{M+1}$. This fact
together with Lemma~\ref{lem:epsCts} proves the proposition.
\end{proof}

\subsection{Relating the LW and self-energy expansions}\label{sec:lwbold}
The bold diagrams for the Luttinger-Ward functional are pinned down in terms of the 
bold diagrams for the self-energy via the following:
\begin{prop}
\label{prop:boldRelation}
If $U$ is a homogeneous quartic polynomial, then for all $k$,
\[
\Phi^{(k)}_G = \frac{1}{2k} \Tr[G \Sigma^{(k)}_G].
\]
\end{prop}

\begin{proof}
Observe that by the transformation rule that for any $G\in \pd$, $\ve, t>0$. 
\[
\Phi[tG,\ve U] = \Phi[G, \ve U\circ (t^{1/2} I)]
\]
Taking the gradient in $G$ of both sides, we have
\[
t \Sigma[tG,\ve U] = \Sigma[G, \ve U\circ (t^{1/2} I)].
\]
Since $U$ is homogeneous quartic, in fact we have 
\[
\Sigma[tG,\ve U] = \frac{1}{t} \Sigma[G, t^2 \ve U].
\]

Then using this relation we compute: 
\begin{eqnarray*}
\Phi[G,\ve U] & = & \int_{0}^{1}\frac{d}{dt}\Phi[tG,\ve U]\,dt\\
 & = & \int_{0}^{1}\Tr[G\Sigma[tG,\ve U]]\,dt\\
 & = & \int_{0}^{1}\frac{1}{t}\Tr[G\Sigma[G,t^{2}\ve U]]\,dt\\
 & = & \int_{0}^{1}\frac{1}{t}\left[\sum_{k=1}^{M}\Tr\left[G\Sigma_{G}^{(k)}\right] t^{2k}\ve^k +O\left(t^{2(M+1)}\ve^{M+1}\right)\right]\,dt\\
 & = & \int_{0}^{1}\left[\sum_{k=1}^{M}\Tr\left[G\Sigma_{G}^{(k)}\right]t^{2k-1}\ve^{k}+O\left(t^{2M+1}\ve^{M+1}\right)\right]\,dt.
\end{eqnarray*}
Now since $t$ ranges from $0$ to $1$ in the integrand, we have
that $t^{2N+1}\ve^{N+1}\leq\ve^{N+1}$, and therefore
\begin{eqnarray*}
\Phi[G,\ve U] & = & \int_{0}^{1}\left[\sum_{k=1}^{M}\Tr\left[G\Sigma_{G}^{(k)}\right]t^{2k-1}\ve^{k}\right]\,dt+O(\ve^{M+1})\\
 & = & \sum_{k=1}^{M}\frac{1}{2k}\Tr\left[G\Sigma_{G}^{(k)}\right]\ve^{k}+O(\ve^{M+1}).
\end{eqnarray*}
This establishes the proposition.
\end{proof}

\subsection{Diagram-free discussion of results from Part I}
\label{sec:p1Review}
For $U$ satisfying the weak growth condition and 
$A \in \dom \Omega[\,\cdot\,,U]$, define 
\[
\sigma[A,U] := A - (G[A,U])^{-1}.
\]
Here we use the lowercase $\sigma$ to emphasize that the self-energy 
here is being considered as a functional of $A$ (not $G$), together with the interaction.

Now we set the notation of $U$ to indicated a fixed generalized Coulomb interaction~\eqref{eqn:Uterm}.
Further define
\begin{equation}
\label{eqn:gsigmaDef}
G_A(\ve) := G[A,\ve U], \quad \sigma_A(\ve) := \sigma[A,\ve U].
\end{equation}

The following lemma concerns the \emph{bare} diagrammatic expansion of the Green's 
function and the self-energy, i.e., the asymptotic 
series for $G_A$ and $\sigma_A$.

\begin{lem}
\label{lem:bareSeries}
For fixed $A \in \pd$, all derivatives 
$\frac{\ud^{n}G_A}{\ud \ve^{n}} : (0,\infty)\ra \pd$ and $\frac{\ud^{n}\sigma_A}{\ud \ve^{n}}: (0,\infty) \ra \symm$ extend continuously 
to $[0,\infty)$. In fact, interpreted as functions of both $A$ and $\ve$, 
$\frac{\ud^{n}G_A}{\ud \ve^{n}}(\ve)$ and $\frac{\ud^{n}\sigma_A}{\ud \ve^{n}}(\ve)$ extend continuously to $\pd\times [0,\infty)$.
 Moreover, 
we have asymptotic series expansions
\[
G_A(\ve) = \sum_{k=0}^\infty g^{(k)}_A \ve^k, \quad \sigma_A(\ve) =
\sum_{k=1}^\infty \sigma^{(k)}_A \ve^k, 
\]
where the coefficient functions $g^{(k)}_A$ and $\sigma^{(k)}_A$ are
polynomials in $A^{-1}$. More precisely, 
$g^{(k)}_A$ and $\sigma^{(k)}_A$ are homogeneous polynomials of degrees $2k+1$ and $2k-1$, respectively. 
(Note that the zeroth-order term $\sigma_A^{(0)}$ is 
implicitly zero.)

Finally, let $G_{A}^{(\le M)}(\ve)$ and $\sigma_{A}^{(\le M)}(\ve)$ denote the $M$-th partial sums of 
the above asymptotic series for $G_A(\ve)$ and $\sigma_A(\ve)$,
respectively. For every 
$A \in \pd$, there exists a neighborhood $\mathcal{N}$ of $A$ in $\pd$ on which the truncation errors can 
actually be bounded
\[
\left\vert G_A(\ve) - G^{(\le M)}_{A}(\ve) \right\vert \leq C \ve^{M+1}, 
\quad
\left\vert \sigma_A(\ve) - \sigma^{(\le M)}_{A}(\ve) \right\vert \leq C \ve^{M+1}
\]
for all $\epsilon \in [0,\tau]$, with $C, \tau$ \emph{independent} of $A \in \mc{N}$. 
\end{lem}
\begin{proof}
The asymptotic series expansions for $G_A$ and $\Sigma_A$ are established 
in Theorems 3.15 and 3.17 of Part I. The continuous extension of the derivatives of 
$G_A$ and $\sigma_A$ to $[0,\infty)$ follows from differentiation under the 
integral and simple dominated convergence arguments.

The uniform error bound follows from a Lagrange error bound argument as in 
Proposition \ref{prop:asymptoticSeries}, together with the continuity of 
$\frac{\ud^{n}G_A}{\ud \ve^{n}}(\ve)$ and $\frac{\ud^{n}\sigma_A}{\ud \ve^{n}}(\ve)$ on $\pd\times [0,\infty)$.
\end{proof}

Inspired by Eq.~\eqref{eqn:thmSeries2}, let 
\[
 \mathbf{S}_G^{(k)} = \sum_{\Gamma_{\mr{s}} \in \mf{F}_2^{\mr{2PI}},\,\mathrm{order}\,k} \frac{\mathbf{F}_{\Gamma_{\mr{s}}}}{S_{\Gamma_{\mr{s}}}}.
\] 
In fact $\mathbf{S}_G^{(k)}$ is polynomial in $G$, homogeneous of degree $2k-1$. 
At this point we do not yet know that $\mathbf{S}_G^{(k)}$ 
coincides with $\Sigma_G^{(k)}$, and indeed this is what we want to show. For any $G$, also define the 
partial sum
\[
 \mathbf{S}_G^{(\leq M)}(\ve) := \sum_{k=1}^M  \mathbf{S}_G^{(k)} \ve^k.
\]
Then the main result 
(Theorem 4.12) of Part I can be phrased as follows. 
\begin{thm}
\label{thm:p1Main}
For any fixed $A\in \pd$, the expressions
\[
\mathbf{S}_{G^{(\le M)}_{A}(\ve)}^{(\le M)}(\ve) = \sum_{k=1}^M
\mathbf{S}_{G^{(\le M)}_{A}(\ve)}^{(k)} \ve^k, \quad  \sigma_{A}^{(\le M)}(\ve) = \sum_{k=1}^M \sigma^{(k)}_A \ve^k
\]
agree as polynomials in $\ve$ up to order $M$, and hence they 
agree as joint polynomials in $(A^{-1},\ve)$ after neglecting all terms in which $\ve$ appears 
degree at least $M+1$.
\end{thm}

\subsection{Derivation of self-energy bold
diagrams}\label{sec:selfenergybold}

We have already shown that there \emph{exist} asymptotic series for the 
LW functional and the self-energy. 
The remainder of Theorem \ref{thm:boldDiag} then consists of identifying that the
self-energy coefficients $\Sigma_G^{(k)}$ are indeed given by the bold diagrammatic expansion, i.e., 
that $\Sigma_G^{(k)} = \mathbf{S}_G^{(k)}$. 
Equivalently, we want to show that the partial sums 
$\mathbf{S}_G^{(\leq M)}(\ve)$ 
and $\Sigma_G^{(\leq M)}(\ve)$, which are polynomials 
of degree $M$ in $\ve$, are equal.
We will think of $G \in \pd$ as fixed throughout the 
following discussion, and 
we omit dependence on $G$ from some of the notation 
below to avoid excess clutter. 
We will also think of 
$M$ as a fixed positive integer and $\ve >0$ as variable (and sufficiently small).

Since our series expansion is only valid in the asymptotic sense, for
any finite $M$ we consider the truncation
\[
\Sigma^{(\le M)}_G (\ve) := \sum_{k=1}^M
\Sigma^{(k)}_G\, \ve^k.
\]
 Then we have $\Sigma_G (\ve) -
\Sigma^{(\le M)}_G (\ve) = O(\ve^{M+1})$. For the purpose of this discussion, $O(\ve^{M+1})$ will be thought of as negligibly small, and `$\approx$' will be used to denote equality up to error $O(\ve^{M+1})$. Meanwhile 
`$\sim$' will be used to denote error that is $O(\ve^{M+1-p})$ for all $p\in(0,1)$, equivalently 
$O(\ve^{M+\delta})$ for all $\delta \in (0,1)$. We remark that the
difference between the relations `$\approx$' and `$\sim$' is due to technical reasons
to be detailed later, and may be neglected on first reading.

Note that it actually suffices to show 
that 
$\Sigma_G^{(\leq M)}(\ve) \sim \mathbf{S}_G^{(\leq M)}(\ve)$. 
Indeed, both sides are polynomials of degree $M$ in $\ve$. Thus their difference is a polynomial 
of degree $\leq M$. If the degree-$n$ part of the difference is nonzero for some $n = 1,\ldots, M$, then the difference is 
not $O(\ve^{n+\delta})$ for any $\delta>0$. But if 
$\Sigma_G^{(\leq M)}(\ve) \sim \mathbf{S}_G^{(\leq M)}(\ve)$, then 
the difference is $O(\ve^{n+\delta})$ for all $n=1,\ldots,M$, $\delta \in (0,1)$. Thus in this case the difference is zero.
With this reduction in mind, we now make a simple yet critical observation, namely that $\Sigma^{(\le M)}_G (\ve)$
can be identified as the \textit{exact} self-energy yielded by a 
modified interaction term. 
This will allow us to identify a quadratic form $A^{(M)}(\ve)$, for which dependence on $G$ has been suppressed 
from the notation, which generates (up to negligible error) the  
Green's function $G$ under the interaction $\ve U$.
\begin{lem}
\label{lem:SEpartialSum}
With notation as in the preceding discussion, $\Sigma^{(\le M)}_G (\ve)$ is the self-energy induced by the interaction ${U}^{(M)}_\ve(x) := \ve U(x) + \frac{1}{2} x^T \left[\Sigma_G(\ve) - \Sigma^{(\le M)}_G (\ve) \right] x $, i.e.,
\[
\Sigma^{(\le M)}_G (\ve) = \Sigma[G,{U}^{(M)}_\ve],
\]
and moreover
\[
A^{(M)} (\ve) := A\left[G,{U}^{(M)}_\ve \right] = G^{-1} + \Sigma^{(\le M)}_G (\ve).
\]
Thus we may identify
\[
G = G[A^{(M)}(\ve),U_\ve^{(M)}],\quad \Sigma^{(\le M)}_G (\ve) = \sigma[A^{(M)} (\ve),{U}^{(M)}_\ve].
\]
\end{lem}
\begin{proof}
Recalling that $A_G(\ve) = A[G,\ve U]$ and 
$\Sigma_G (\ve) = \Sigma[G,\ve U]$, write 
\begin{eqnarray*}
\frac12 x^T A_G(\ve) x + U(x) & = & \frac12 x^T \left( A_G(\ve) - \Sigma_G(\ve) + \Sigma^{(\le M)}_G (\ve) \right) x + U_\ve ^{(M)} (x) \\
& = &  \frac12 x^T \left( G^{-1} + \Sigma^{(\le M)}_G (\ve) \right) x + U_\ve ^{(M)} (x).
\end{eqnarray*}
It follows that under the interaction ${U}^{(M)}_\ve$, the quadratic form $G^{-1} + \Sigma^{(\le M)}_G (\ve)$ corresponds to the (interacting) Green's function $G$. This establishes the second statement of the lemma, i.e., that \[
A[G,{U}^{(M)}_\ve] = G^{-1} + \Sigma^{(\le M)}_G (\ve).
\]
Moreover, by the Dyson equation we have that 
\[
\Sigma[G,{U}^{(M)}_\ve] = A[G,{U}^{(M)}_\ve] - G^{-1} = \Sigma^{(\le M)}_G (\ve),
\]
which is the first statement of the lemma. The last 
statement then follows from the second, together 
with the definitions of $G[\,\cdot\,,\,\cdot\,]$ and $\sigma[\,\cdot\,,\,\cdot\,]$.
\end{proof}

\vspace{2mm}

\begin{rem}
Note carefully that Lemma \ref{lem:SEpartialSum} is a
non-perturbative fact and is valid for all $\ve>0$, though we shall apply it in a 
perturbative context.
\end{rem}

\vspace{2mm}

At this point we have defined the terms needed 
to present a schematic diagram (Figure \ref{fig:schematic}) of our proof that 
$\Sigma_{G}^{(\le M)}(\ve) \sim \mathbf{S}^{(\le M)}_G (\ve)$. Although the motivation for this schematic 
may not be fully clear at this point, the reader should 
refer back to it as needed for perspective.
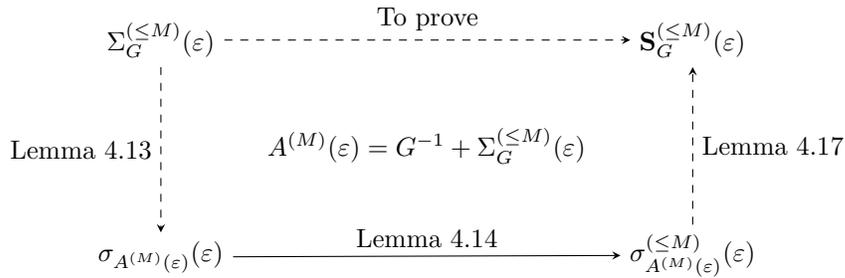
\begin{figure}[h]
\begin{center}
\begin{tikzpicture}
  \matrix (m) [matrix of math nodes,row sep=2em,column sep=1em,minimum width=2em] {
  \Sigma_{G}^{(\le M)}(\ve) & & \mathbf{S}^{(\le M)}_G (\ve) \\
  & A^{(M)}(\ve) = G^{-1} + \Sigma_{G}^{(\le M)}(\ve) & \\
  \sigma_{A^{(M)}(\ve)}(\ve) & & \sigma^{(\le M)}_{A^{(M)}(\ve)} (\ve) \\};
  \path[-stealth]
    (m-1-1) edge [dashed] node [left] {Lemma~\ref{lem:SigmaSim}}(m-3-1)
            edge [dashed] node [above] {To prove}(m-1-3)
    (m-3-3) edge [dashed] node [right] {Lemma \ref{lem:endgame}} (m-1-3)
    (m-3-1) edge          node [above] {Lemma~\ref{lem:SigmaSim2}}(m-3-3);
\end{tikzpicture}
\end{center}
\caption{Schematic diagram for proving the bold diagrammatic expansion.
Dashed lines indicate `$\sim$', and solid lines indicate `$\approx$'.}
\label{fig:schematic}
\end{figure}

Now recalling the definitions \eqref{eqn:gsigmaDef}, we can write
\begin{equation}
\label{eqn:boldProp}
G_{A^{(M)}(\ve)}(\ve) = G[A^{(M)}(\ve) , \ve U], \quad \sigma_{A^{(M)}(\ve)} (\ve) := \sigma[A^{(M)}(\ve) , \ve U].
\end{equation}
Meanwhile, following Lemma \ref{lem:SEpartialSum} we have the identities
\begin{equation}
\label{eqn:boldProp2}
G = G[A^{(M)}(\ve) , U_\ve^{(M)}], \quad \Sigma_G^{(\le M)} (\ve) = \sigma[A^{(M)}(\ve) , U_\ve^{(M)}].
\end{equation}
Note that pointwise, $\ve U$ and $U_\ve^{(M)}$ differ negligibly, but the form of $\ve U$ is 
simpler and easier to work with going forward. 

Based on Eqs.~\eqref{eqn:boldProp} 
and \eqref{eqn:boldProp2}, one then hopes that $G_{A^{(M)}(\ve)}(\ve)$ is close to $G$ 
and $\sigma_{A^{(M)}(\ve)}(\ve)$ is close to $\Sigma_G^{(\leq M)}(\ve)$. This is the content 
of the next two lemmas.

\vspace{2mm}

\begin{lem}
\label{lem:Gsim}
$G_{A^{(M)}(\ve)}(\ve) \sim G$.
\end{lem}

\begin{proof}
See appendix \ref{app:Gsim}.
\end{proof}

\vspace{2mm}

\begin{lem}
\label{lem:SigmaSim}
$\sigma_{A^{(M)}(\ve)}(\ve) \sim \Sigma^{(\le M)}_{G}(\ve)$.
\end{lem}

\begin{proof}
Based on Eqs.~\eqref{eqn:boldProp} 
and \eqref{eqn:boldProp2}, we want to show that $ \sigma[A^{(M)}(\ve), U_\ve^{(M)}] \sim  \sigma[A^{(M)}(\ve), \ve U]$. 
We have already shown that $ G = G[A^{(M)}(\ve), U_\ve^{(M)}] \sim  G[A^{(M)}(\ve), \ve U]$, from which it follows that 
\[
A^{(M)}(\ve) - (G[A^{(M)}(\ve), U_\ve^{(M)}])^{-1} \sim  A^{(M)}(\ve) - (G[A^{(M)}(\ve), \ve U])^{-1},
\]
which is exactly what we want to show.
\end{proof}

\vspace{2mm}

Then we can use $\sigma_{A^{(M)}(\ve)}(\ve)$ as a stepping stone 
to relate $\Sigma^{(\leq M)}_G (\ve)$ with the bare diagrammatic expansion 
for the self-energy via the following:

\vspace{2mm}

\begin{lem}
\label{lem:SigmaSim2}
$\sigma_{A^{(M)}(\ve)}(\ve) \approx 
\sigma_{A^{(M)}(\ve)}^{(\leq M)}(\ve)$
\end{lem}
\begin{proof}
Since $A^{(M)}(\ve) = G^{-1} + O(\ve)$, the result 
follows from Lemma \ref{lem:bareSeries} (in particular, 
the locally uniform bound on truncation error of the bare self-energy series).
\end{proof}

\vspace{2mm}

We can prove a similar fact (which will be useful later on) regarding the bare series for the interacting Green's function:

\vspace{2mm}

\begin{lem}
\label{lem:Gsim2}
$G_{A^{(M)}(\ve)} (\ve) \approx G_{A^{(M)}(\ve)}^{(\leq M)} (\ve)$.
\end{lem}
\begin{proof}
Since $A^{(M)}(\ve) = G^{-1} + O(\ve)$, the result 
follows from Lemma \ref{lem:bareSeries} (in particular, 
the locally uniform bound on truncation error of the bare series for the interacting Green's function).
\end{proof}

\vspace{2mm}

From Lemmas \ref{lem:Gsim} and \ref{lem:Gsim2} 
we immediately obtain:

\vspace{2mm}

\begin{lem}
\label{lem:Gsim3}
$G_{A^{(M)}(\ve)}^{(\leq M)} (\ve) \sim G$.
\end{lem}

\vspace{2mm}

Finally, we are ready to state and prove the last leg 
of the schematic diagram (Figure \ref{fig:schematic}):

\vspace{2mm}

\begin{lem}
\label{lem:endgame}
$\mathbf{S}^{(\le M)}_G \sim \sigma_{A^{(M)}(\ve)}^{(\leq M)}(\ve)$.
\end{lem}

\begin{proof}
Consider $\mathbf{S}_{G_{A}^{(\le M)}}^{(\le M)}$ as a polynomial in
$(A^{-1},\ve)$, 
and let $P(A^{-1},\ve)$ be the contribution of terms in which $\ve$ appears with degree at 
least $M+1$. By Theorem \ref{thm:p1Main} we have the equality 
\[
\mathbf{S}_{G_{A}^{(\le M)}(\ve)}^{(\le M)}(\ve) - P(A^{-1},\ve) 
= \sigma_A^{(\leq M)} (\ve)
\]
of polynomials in $(A^{-1},\ve)$.  Then substituting
$A\leftarrow A^{(M)}(\ve)$ we obtain
\begin{equation}
\label{eqn:polynomial}
\mathbf{S}_{G_{A^{(M)}(\ve)}^{(\le M)}(\ve)}^{(\le M)}(\ve) - P([A^{(M)}(\ve)]^{-1},\ve) 
= \sigma_{A^{(M)}(\ve)}^{(\leq M)} (\ve).
\end{equation}
Although the first term on the left-hand side of 
Eq.~\eqref{eqn:polynomial} looks quite intimidating, we 
can recognize it as $\mathbf{S}_{ \mathbf{G}(\ve)  }^{(\leq M)}(\ve)$, where  
\[
\mathbf{G}(\ve) := G_{A^{(M)}(\ve)}^{(\le M)}(\ve) \sim G
\]
is the expression from Lemma \ref{lem:Gsim3}.
Since 
$\mathbf{S}_{ [\,\cdot\,]  }^{(\leq M)}(\ve) = 
\sum_{k=1}^M \mathbf{S}_{ [\,\cdot\,]  }^{(k)} \ve^k
$, where each $\mathbf{S}_{ [\,\cdot\,]  }^{(k)}$ is a polynomial (homogeneous of positive degree) in the subscript 
slot, it follows that 
\[
\mathbf{S}_{ \mathbf{G}(\ve)  }^{(\leq M)}(\ve) \sim \mathbf{S}_{ G }^{(\leq M)}(\ve).
\]
Then from Eq.~\eqref{eqn:polynomial} we obtain 
\[
\mathbf{S}_{G}^{(\le M)}(\ve) - P([A^{(M)}(\ve)]^{-1},\ve) 
\sim \sigma_{A^{(M)}(\ve)}^{(\leq M)} (\ve).
\]
But since $[A^{(M)}(\ve)]^{-1} = G + O(\ve)$ and since $P$ only includes 
terms of degree at least $M+1$ in the second slot, it 
follows that $P([A^{(M)}(\ve)]^{-1},\ve) \approx 0$, and 
the desired result follows.
\end{proof}

Taken together (as indicated in Figure \ref{fig:schematic}), 
Lemmas \ref{lem:SigmaSim}, \ref{lem:SigmaSim2}, and 
\ref{lem:endgame} 
imply that
$\Sigma_{G}^{(\le M)}(\ve) \sim \mathbf{S}^{(\le M)}_G (\ve)$
as desired, and the proof of Theorem \ref{thm:boldDiag} is complete.

\subsection{Caveat concerning truncation of the bold diagrammatic expansion}\label{sec:failbold}
Although the LW and self-energy functionals are defined even for $G$ such that
the corresponding quadratic form $A = A[G]$ is indefinite (and hence there is no physical 
bare non-interacting Green's function), Green's function methods 
(as discussed in section 4.7 of Part I) based on truncation of the bold diagrammatic 
expansion can fail dramatically in the case of indefinite $A$. 
One can encounter divergent behavior as the interaction becomes small, or the 
Green's function method may fail to admit a solution.
Both failure modes can demonstrated by simple one-dimensional
examples. The relevance of these to the solution of the quantum many-body problem is 
at this point unclear.

Consider
the one-dimensional example of
\begin{equation}
  Z = \int_{\RR} e^{\frac12  x^2 - \frac18 \lambda x^4}\ud x, 
  \label{eqn:partition1Dneg}
\end{equation}
where $a = -1$.
The corresponding non-interacting Green's function is $G^{0}=-1<0$ and
hence is not even a physical Green's function. 

Nonetheless with
$\lambda>0$ the true Green's function is still well-defined via
\[
  G = \frac{1}{Z} \int_{\RR} x^2 e^{\frac12  x^2 - \frac18 \lambda x^4}\ud x.
  \label{}
\]

We now compute $G$ via the Hartree-Fock method (cf. section 4.7 of Part I), 
i.e., we approximate 
the self-energy as
\[
  \Sigma^{(1)} = -\frac12 \lambda G - \lambda G = -\frac32 \lambda G. 
  \label{}
\]
Hence the self-consistent solution $G^{(1)}$ of the Dyson equation solves 
\[
  \frac{1}{G^{(1)}} = -1 + \frac32 \lambda G^{(1)}.
  \label{}
\]
There is only one positive (physical) solution to this equation, namely
\[
  G^{(1)} = \frac{1 + \sqrt{1+6\lambda}}{3\lambda}.
  \label{eqn:dyson1d1}
\]

In the spirit of perturbation theory, one might hope that
$G^{(1)}$ is a good approximation to $G$ at least when $\lambda\to
0$.
However we see just the opposite. 
This is perhaps not surprising because the exact 
Green's function $G$ itself blows up in this limit.

The failure of the method as $\lambda\to 0$
can be understood more precisely as follows. Rewrite the
Hamiltonian from~\eqref{eqn:partition1Dneg} as
\[
  \frac18 \lambda
  \left(x^2-\frac{2}{\lambda}\right)^2 - \frac{1}{2\lambda}.
  \label{}
\]
The corresponding Gibbs measure (which is unaffected by the additive 
constant) then concentrates about two peaks at 
$x=\pm \sqrt{\frac{2}{\lambda}}$ as $\lambda \rightarrow 0$.
Hence we expect 
\[
  G \sim 2 \lambda^{-1}.
\]
We note that, in contrast with the statement of 
Lemma~\ref{lem:bareSeries}, the limit $\lim_{\lambda\to 0+} G(\lambda)$
does not exist. 
According to Eq.~\eqref{eqn:dyson1d1} 
\[
G^{(1)} \sim \frac{2}{3} \lambda^{-1}.
\]
We find that as $\lambda\to 0+$, $G$ and its first order approximation
$G^{(1)}$ do not agree.

If we include the second-order terms of the bold diagrammatic expansion 
\begin{equation}
  \Sigma^{(2)} = \frac12 \lambda^2 G^3 + \lambda^2 G^3 = \frac32
  \lambda^2 G^3. 
  \label{}
\end{equation}
Then the self-consistent solution $G^{(2)}$ of the Dyson equation solves
\begin{equation*}
  \frac{1}{G^{(2)}} = -1 + \frac32 \lambda G^{(2)} - \frac32 \lambda^2
  \left(G^{(2)}\right)^3.
  \label{eqn:dyson1d2}
\end{equation*}
This yields a quartic equation in the scalar $G^{(2)}$, which in fact has no solution for physical $G^{(2)}$, i.e., $G^{(2)} > 0$.

To see this, first ease the notation by substituting $x \leftarrow G^{(2)}$, so we are interested 
in the solutions $x>0$ of 
\[
\frac32 \left[ (\lambda^{1/2}x)^4 - (\lambda^{1/2} x)^2 \right] + x + 1 = 0. 
\]
But $y^4 - y^2 \geq -\frac14$ for all $y$, 
so the first term is at least $-\frac38$, which evidently implies that no solutions exist for $x>0$.

\section{Proof of the continuous extension of the LW functional}
\label{sec:boundary}
In section \ref{sec:continuousState} we motivated the
continuous extension of the LW functional to the boundary of $\pd$ and 
stated this result in two equivalent forms (Theorems \ref{thm:contUpTo} and 
\ref{thm:contUpToSpecialCase}).
In this section we prove the continuous extension property (for interactions of
 strong growth). We also develop  
the counterexample promised earlier, an interaction of weak but not strong growth 
for which the continuous extension property fails.

The section is outlined as follows.
In section \ref{subsec:setup}, 
we describe some preliminary reductions in the proof of the continuous 
extension property, after which the proof can be divided into two parts: 
lower-bounding the limit inferior of the LW functional as 
the argument approaches the boundary and upper-bounding the limit supremum. In section \ref{subsec:lowerbound}, we prove the lower bound, and in section \ref{subsec:upperbound} we prove the upper bound. In section 
\ref{subsec:alternateView} we provide an alternate view on the continuous extension property from the 
Legendre dual side, and in section \ref{subsec:counterexample} we use this perspective to exhibit the 
aforementioned counterexample to the continuous extension property, which satisfies the weak growth condition but not the strong one.

\subsection{Proof setup}
\label{subsec:setup}
We are going to prove Theorem \ref{thm:contUpToSpecialCase},
which as we have remarked suffices to prove Theorem \ref{thm:contUpTo}
by changing coordinates via Proposition
\ref{prop:LWtransformation}.

Suppose $G\in \psd$ is of the form
\[
G=\left(\begin{array}{cc}
G_{p} & 0\\
0 & 0
\end{array}\right),
\]
where $G_{p}\in\mathcal{S}_{++}^{p}$, and suppose that $G^{(j)} \in \pd$ with $G^{(j)}\ra G$ as $j\ra\infty$. For each $j$, diagonalize $G^{(j)}=\sum_{i=1}^{N}\lambda_{i}^{(j)} v_{i}^{(j)} \left( v_{i}^{(j)} \right) ^{T}$,
where the $v_{i}^{(j)}$ are orthonormal, $\lambda_{i}^{(j)}>0$ for $i=1,\ldots,N$.

We want to show that 
\[
\Phi_{n}[G^{(j)},U] \ra \Phi_{p}[G_p,U(\,\cdot\,,0)].
\]
It suffices to show that every subsequence has a convergent subsequence
with its limit being $\Phi_{p}[G_p,U(\,\cdot\,,0)]$. The $G^{(j)}$ are
convergent, hence bounded (in the $\Vert \cdot \Vert_2$ norm), so the
$\lambda_i^{(j)}$ are bounded. Moreover, the $v_i^{(j)}$ are all of unit
length, hence bounded, so by passing to a subsequence if necessary we
can assume that, for each $i$, there exist $\lambda_i, v_i$ such that
$\lambda_i^{(j)} \ra \lambda_i$ and $v_i^{(j)} \ra v_i$ as $j\ra
\infty$. It follows that the $v_i$ are orthonormal and that $G$ can be
diagonalized as $G=\sum_{i=1}^{N}\lambda_{i}v_{i}v_{i}^{T}$. Since $G_p$
is positive definite, we must have $\lambda_{i}>0$ for $i=1,\ldots,p$,
and moreover $\lambda_{i}=0$ for $i=p+1,\ldots,N$. Evidently, the
eigenvectors of $G$ with strictly positive eigenvalues must be precisely
the eigenvectors of $G_p$, concatenated with $N-p$ zero entries, i.e.,
for $i=1,\ldots,p$, $v_i$ must be of the form $(*,0)$. By orthogonality,
for $i=p+1,\ldots, n$, $v_i$ must be of the form $(0,*)$.

For convenience we also establish the following notation:  
\[V_{G}:= \mathrm{span}\{v_1 ,\ldots, v_p\},\  V_{G^{(j)}}:= \mathrm{span}\{v_1^{(j)} ,\ldots, v_p^{(j)}\}.
\]

Now the proof consists of proving two bounds: a lower bound
\[
\liminf_{j\ra\infty}\Phi_{N}[G^{(j)},U] \geq  \Phi_{p}[G_{p},U(\cdot,0)]
 \]
 and an upper bound
 \[
\limsup_{j\ra\infty}\Phi_{N}[G^{(j)},U] \leq  \Phi_{p}[G_{p},U(\cdot,0)].
 \]
 These bounds will be proved in the next two sections, i.e., sections \ref{subsec:lowerbound} and \ref{subsec:upperbound}, respectively.

\subsection{Lower bound}
\label{subsec:lowerbound}

We want to establish a lower bound on $\Phi_{N}[G_{j},U]$
via our expression for $\mathcal{F}_{N}$ as a supremum: 
\begin{equation}
\label{eq:FNsup}
\mathcal{F}_{N}[G^{(j)},U]=\sup_{\mu\in \mathcal{G}_{N}^{-1}(G^{(j)})}\left[H(\mu)-\int U\,\ud\mu\right].
\end{equation}
This strategy requires us to construct measures $\mu^{(j)} \in
\mathcal{G}_{N}^{-1}(G^{(j)})$. Intuitively, what one hopes to do (though this 
strategy will require some modification) is the
following: consider the measure $\alpha$ on $\R^p$ that attains the
supremum in the analogous expression for $\mathcal{F}_p
[G_p,U(\,\cdot\,,0)]$, identify this measure with a measure on $V_G
\simeq \R^p$, rotate and scale appropriately to obtain a measure
$\alpha^{(j)}$ supported on $V_{G^{(j)}}$ with the correct second-order
moments with respect to this subspace, and finally take the direct sum
with an appropriate Gaussian measure $\beta^{(j)}$ on
$V_{G^{(j)}}^\perp$. Unfortunately, due to difficulties of analysis, it is not clear 
how to then prove
the desired limit as $j\ra \infty$.

However, the analysis of this limit would be 
feasible if the $\mu^{(j)}$
had compact support (which they evidently do not).
Then our approach is to carry out a construction that preserves the spirit of the `ideal' construction just described but instead works with $\mu^{(j)}$ of (uniform) compact support.

For convenience we let $\mathcal{M}_c \subset
\mathcal{M}_2$ denote the subset of measures of compact support.
The acceptability of working with measures of compact support can be motivated by the following lemma, which
will be used below. 
(In the statement we temporarily suppress dependence on the interaction and the
dimension from the notation.)
\begin{lem}
\label{lem:Ftech}
For all $G\in \symm$,
\[
\mathcal{F} [G] = \sup_{\mu \in \mathcal{G}^{-1}(G)\cap \mathcal{M}_c} \left[H(\mu) - \int U\,\ud\mu  \right].
\]
\end{lem}

Now we outline our \emph{actual} construction of the $\mu^{(j)}$. Consider an
\textit{arbitrary} measure $\alpha\in \mathcal{G}_{p}^{-1}(G_{p})$ with
compact support on $\R^{p}\simeq V_{G}$. (We abuse notation slightly by
considering $\alpha$ as a measure on both $\R^{p}$ and $V_{G}$.) The
idea now is to construct a measure in $\mu^{(j)} \in
\mathcal{G}_{N}^{-1}(G^{(j)})$ by rotating $\alpha$ and scaling
appropriately to obtain a measure $\alpha^{(j)}$  supported on
$V_{G^{(j)}}$ and then taking the direct sum with a compactly supported
measure $\beta^{(j)}$ on $V_{G^{(j)}}^\perp$ (the details of which will
be discussed later). In fact the supremum in \eqref{eq:FNsup} will be approximately 
attained by a measure of this form as $j\ra \infty$, i.e., our lower bound will be tight as $j\ra \infty$.

 Accordingly, for the construction of $\alpha^{(j)}$, let $O^{(j)}$ be the orthogonal linear transformation sending $v_i \mapsto v_i^{(j)}$, and let $D^{(j)}$ be the linear transformation with matrix (in the $v_i^{(j)}$ basis) given by
\[
\mathrm{diag}\left( \sqrt{\lambda_{1}^{(j)}/\lambda_{1}},\ldots,\sqrt{\lambda_{p}^{(j)}/\lambda_{p}}, 1,\ldots,1 \right).
\]
Then define $T^{(j)} := D^{(j)} O^{(j)}$ and $\alpha^{(j)} :=  T^{(j)}
\# \alpha$. Note that $T^{(j)} \ra I_n$ as $j \ra \infty$.
Moreover, observe that $\alpha^{(j)}$ is a measure supported on $V_{G^{(j)}}$ with second-order moment matrix 
given by $\mathrm{diag}(\lambda_{1}^{(j)},\ldots,\lambda_{p}^{(j)})$ with
respect to the coordinates on $V_{G^{(j)}}$ induced by the orthonormal
basis $v_{1}^{(j)},\ldots,v_{p}^{(j)}$.

Now we turn to the construction of $\beta^{(j)}$. Let $R>1$ and let $\gamma$ be a measure supported on $[-R,R]$ with $\int x^2 \,d\gamma = 1$. The parameter $R$ will control the size of the support of $\beta^{(j)}$ and will be sent to $+\infty$ at the very end of the proof of the lower bound (after the limit in $j$ has been taken).
Then define
\[
\Lambda^{(j)} :=
\mathrm{diag}\left( \sqrt{\lambda_{p+1}^{(j)}},\ldots,\sqrt{\lambda_{N}^{(j)}} \right),
\]
and define a measure $\beta^{(j)}$ on $\R^{N-p}$ by $\beta^{(j)} := \Lambda^{(j)} \# (\gamma \times \cdots \times \gamma)$. Note that $\Lambda^{(j)}\ra 0$ as $j\ra\infty$. Abusing notation slightly, we will also identify $\beta^{(j)}$ with a measure supported on $V_{G^{(j)}}^{\perp}\simeq\R^{N-p}$ via the identification of the orthonormal basis $v_{p+1}^{(j)},\ldots,v_{N}^{(j)}$ for $V_{G^{(j)}}^\perp$ with the standard basis of $\R^{N-p}$.

Finally, define the product measure $\mu^{(j)}:=\alpha^{(j)}\times\beta^{(j)}$
with respect to the product structure $\Rn=V_{G^{(j)}}\times V_{G^{(j)}}^{\perp}$,
and note that $\mu^{(j)}\in \mathcal{G}_N^{-1}(G^{(j)})$, so by \eqref{eq:FNsup}, 
\begin{eqnarray*}
\mathcal{F}_{N}[G^{(j)},U] & \geq & H(\alpha^{(j)}\times\beta^{(j)})-\int U\,\ud\mu^{(j)}\\
 & = & H(\alpha^{(j)})+H(\beta^{(j)})-\int U\,d \mu^{(j)}\\
 & = & H(\alpha)-\int U\,\ud\mu^{(j)}+\frac{1}{2}\sum_{i=p+1}^{N}\log\lambda_{i}^{(j)}+(N-p)H(\gamma),
\end{eqnarray*}
 where $H(\alpha^{(j)})$ and $H(\beta^{(j)})$ are the entropies
of $\alpha^{(j)}$ and $\beta^{(j)}$ on the probability spaces $V_{G^{(j)}}$
and $V_{G^{(j)}}^{\perp}$, respectively.

Notice that there is a compact set on which $\emph{all}$ of the measures $\mu^{(j)}$ are supported. It is then not difficult to see that $\mu^{(j)}$ converges weakly to the measure $\alpha \times \delta_0$, where the product is with respect to the product structure $\Rn = V_G \times V_G^\perp$ and $\delta_0$ is the Dirac delta measure localized at the origin. By the continuity of $U$ and the uniform boundedness of the supports of $\mu^{(j)}$, this is enough to guarantee that 
\[
\int U\,\ud\mu^{(j)} \ra \int U\,d(\alpha\times\delta_0) = \int U(\cdot,0)\,d\alpha
\]
as $j\ra\infty$.

Next we write the Luttinger-Ward functional in terms of $\mathcal{F}_N$:
\begin{eqnarray*}
\frac{1}{2}\Phi_{N}[G^{(j)},U] & = & \mathcal{\mathcal{F}}_{N}[G^{(j)},U]-\frac{1}{2}\mathrm{Tr}[\log(G^{(j)})]-\frac{N}{2}\log(2\pi e)\\
 & = & \mathcal{\mathcal{F}}_{N}[G^{(j)},U]-\frac{1}{2}\sum_{i=1}^{N}\log\lambda_{i}^{(j)}-\frac{N}{2}\log(2\pi e).
\end{eqnarray*}
 Then combining the preceding observations yields
\begin{eqnarray*}
\liminf_{j\ra\infty}\frac{1}{2}\Phi_{N}[G^{(j)},U] & \geq & \liminf_{j\ra\infty}\left[H(\alpha)-\int U\,\ud\mu^{(j)}-\frac{1}{2}\sum_{i=1}^{p}\log\lambda_{i}^{(j)}-\frac{N}{2}\log(2\pi e)+(N-p)H(\gamma) \right]\\
 & = & H(\alpha)-\int U(\cdot,0)\,d\alpha -\frac{1}{2}\sum_{i=1}^{p}\log\lambda_{i}-\frac{N}{2}\log(2\pi e)+(N-p)H(\gamma)\\
 & = & H(\alpha)-\int U(\cdot,0)\,d\alpha-\frac{1}{2}\mathrm{Tr}\left[\log(G_{p})\right]-\frac{N}{2}\log(2\pi e)+(N-p)H(\gamma).
\end{eqnarray*}

Now for any $\ve > 0$, we can choose $R$ sufficiently large and $\gamma$ supported on $[-R,R]$ such that $H(\gamma) \geq \frac{1}{2} \log(2\pi e) - \ve$. Indeed, note that $\frac{1}{2} \log(2\pi e)$ is the entropy of the standard normal distribution, i.e., the maximal entropy over measures of unit variance. By restricting the normal distribution to $[-R,R]$ for $R$ sufficiently large, we can become arbitrarily close to saturating this bound. Therefore we have that 
\[
\liminf_{j\ra\infty}\frac{1}{2}\Phi_{N}[G^{(j)},U]  \geq 
H(\alpha)-\int U(\cdot,0)\,d\alpha-\frac{1}{2}\mathrm{Tr}\left[\log(G_{p})\right]-\frac{p}{2}\log(2\pi e).
\]
Since $\alpha$ was arbitrary in $\mathcal{G}_{p}^{-1}(G_{p})\cap \mathcal{M}_c$, this establishes
the desired upper bound 
\begin{eqnarray*}
\frac{1}{2}\liminf_{j\ra\infty}\Phi_{N}[G^{(j)},U] & \geq &
\sup_{\alpha\in \mathcal{G}_{p}^{-1}(G_{p})\cap\mathcal{M}_c}\left[H(\alpha)-\int U(\cdot,0)\,d\alpha\right]-\frac{1}{2}\mathrm{Tr}\left[\log(G_{p})\right]-\frac{p}{2}\log(2\pi e)\\
 & = & \frac{1}{2} \Phi_{p}[G_{p},U(\cdot,0)],
\end{eqnarray*}
where we have used Lemma \ref{lem:Ftech}, which allows us to look at the supremum over compactly supported measures.
 
Observe that the proof of the lower bound did not require the strong growth assumption, hence the semi-continuity claim of Remark \ref{rem:growthSC}.

\subsection{Upper bound}
\label{subsec:upperbound}

Next we turn to establishing an upper bound. The basic strategy is to select measures $\mu^{(j)}$ that (approximately) attain the supremum in \eqref{eq:FNsup} and take a limit as $j\ra \infty$.

Before proceeding, let $\ve >0$. Moreover, define $\pi_1$ to be the orthogonal projection onto $V_{G} \simeq \R^p$, and define $\pi_2$ to be the orthogonal projection onto $V_{G}^\perp \simeq \R^{N-p}$.

Now for every $j$, as suggested above choose $\mu^{(j)} \in
\mathcal{G}_N^{-1}(G^{(j)})$ such that 
\[
\mathcal{F}_N [G^{(j)},U] \leq H( \mu^{(j)} ) - \int U \,\ud\mu^{(j)} + \ve.
\]
Therefore
\begin{equation}
\label{eq:limsupPhi}
\Phi_N [G^{(j)},U] \leq \underbrace{H( \mu^{(j)} ) - \int U \,\ud\mu^{(j)} - \frac{1}{2} \sum_{i=1}^N \log(2\pi e \lambda_i^{(j)})}_{=: a_j} \,+ \ \ve.
\end{equation}
Then choose a subsequence $j_k$ such that $\lim_{k\ra\infty} a_{j_k} = \limsup_{j\ra\infty} a_j$.

Now the $\mu^{(j)}$ have uniformly bounded second moments, so by
Markov's inequality, the sequence $\mu^{(j)}$ is tight. Then by
Prokhorov's theorem (Theorem \ref{thm:prokhorov}), we can assume, by extracting a further subsequence if necessary, that $\mu^{(j_k)}$ converges weakly to some measure $\mu$.

We claim that $\mathcal{G}_N(\mu) \preceq G$ (so in particular, $\mu \in \mathcal{M}_2$). Indeed, for any $z \in \Rn$, by the Portmanteau theorem for weak convergence of measures (Theorem \ref{thm:portmanteau}) we have
\begin{eqnarray*}
\int(z^T x)^2 \,\ud\mu & \leq & \liminf_{k\ra \infty} \int(z^T x)^2 \,\ud\mu^{(j_k)} \\
& = & \liminf_{k\ra \infty} \int z^T xx^T z \,\ud\mu^{(j_k)}
 =  \liminf_{k\ra \infty} z^T G^{(j_k)} z = z^T G z.
\end{eqnarray*}
It follows that $\mu \in \mathcal{M}_2$ and moreover $z^T \mathcal{G}_n(\mu) z \leq z^T G z$ for all $z$, i.e., $\mc{G}_n(\mu) \preceq G$. In particular, $\mu$ is supported on $V_G$.

Define $T^{(j)}$ to be the \textit{orthogonal} transformation that sends $v^{(j)}_i \mapsto v_i$, so $T^{(j)} \ra I_n$ as $j\ra \infty$. Define $\nu^{(j)}:= T^{(j)}\#\mu^{(j)}$. Again by Prokhorov's theorem, we can assume that $\nu^{(j_k)}$ converges weakly to some measure $\nu$. In fact, we must have $\nu = \mu$. To see this, note that for any continuous compactly supported function  $\phi$ on $\Rn$, we have that $\phi \circ T^{(j)} \ra \phi$ uniformly as $j\ra\infty$. Therefore 
\[
\lim_{j\ra\infty} \int \left\vert \phi - \phi\circ T^{(j)} \right\vert\,\ud\mu^{(j)} \ra 0.
\]
Consequently 
\[
\int \phi\,\ud\mu = \lim_{k\ra\infty} \int\phi \,\ud\mu^{(j_k)} = \lim_{k\ra\infty} \int\phi\circ T^{(j_k)} \,\ud\mu^{(j_k)} = \lim_{k\ra\infty} \int\phi \, d\nu^{(j_k)} = \int\phi\,d\nu.
\]
Since $\mu$ and $\nu$ agree on all continuous compactly supported functions, they must be equal (Riesz representation theorem), and $\nu^{(j_k)} \ra \mu$ weakly.

Define $\mu^{(j)}_i := \pi_i \# \nu^{(j)} = \left(\pi_i \circ T^{(j)}\right)\# \mu^{(j)}$ and $\mu_i :=  \pi_i \#\mu$ for $i=1,2$. It follows  that $\mu^{(j_k)}_i \ra \mu_i$ weakly. Notice (using Fact \ref{fact:productEntropy}) that 
\[
H(\mu^{(j)}) = H(\nu^{(j)}) \leq H(\mu^{(j)}_1) + H(\mu^{(j)}_2) \leq  H(\mu^{(j)}_1) + \frac{1}{2}\sum_{i=p+1}^N \log (2\pi e \lambda_i ^{(j)}).
\]
Therefore, using Lemma \ref{lem:entropyUSC} with the weak convergence $\mu_1^{(j_k)} \ra \mu_1$, we obtain
\begin{eqnarray*}
\lim_{k\ra\infty} a_{j_k} & = & \lim_{k\ra \infty } \left[ H(\mu^{(j_k)}) - \int U \,\ud\mu^{(j_k)} - \frac{1}{2} \sum_{i=1}^N \log(2\pi e \lambda_i^{(j_k)}) \right] \\
&\leq & \limsup_{k\ra \infty } \left[ H(\mu_1^{(j_k)}) - \frac{1}{2}\sum_{i=1}^p \log (2\pi e \lambda_i ^{(j)}) \right] - \liminf_{k\ra \infty } \left[ \int U \,\ud\mu^{(j_k)} \right] \\
& \leq & H(\mu_1) - \liminf_{k\ra \infty } \left[ \int U \,\ud\mu^{(j_k)} \right] - \frac{1}{2} \log( (2\pi e)^p \det G_p).
\end{eqnarray*}

Now for any $\alpha\in \R$, define $U_\alpha(x) = U(x) - \alpha\Vert x\Vert ^2$. Then
\[
\int U \,\ud\mu^{(j)} = \int U_\alpha \,\ud\mu^{(j)} + \alpha \Tr[G^{(j)}].
\]
The utility of this manipulation will be made clear later. By the strong growth condition, $U_\alpha$ is bounded below. Therefore, by the Portmanteau theorem for weak convergence of measures,
\[
\liminf_{k\ra \infty } \left[ \int U \,\ud\mu^{(j_k)} \right]  =  \alpha \Tr[G] +  \liminf_{k\ra \infty } \left[ \int U_\alpha \,\ud\mu^{(j_k)} \right] 
 \geq  \alpha \Tr[G_p] +   \int U_\alpha \,\ud\mu.
\]
Since $\mu$ is supported on $V_G$, in fact 
\[
\int U_\alpha\,\ud\mu = \int U_\alpha(\,\cdot\,,0)\,\ud\mu_1 = \int U(\,\cdot\,,0)\,\ud\mu_1 - \alpha \Tr[\mc{G}_p(\mu_1)],
\]
 and therefore 
\begin{eqnarray*}
\lim_{k\ra\infty} a_{j_k} & \leq & H(\mu_1)-  \int U(\,\cdot\,,0) \,\ud\mu_1 - \frac{1}{2} \log( (2\pi e)^p \det G_p) + \alpha \Tr[\mc{G}_p(\mu_1) - G_p] \\
& \leq & \mathcal{F}_p[\mathcal{G}_p (\mu_1), U(\,\cdot\,,0)] - \frac{1}{2} \log( (2\pi e)^p \det G_p) + \alpha \Tr[\mc{G}_p(\mu_1) - G_p].
\end{eqnarray*}
Recall from \eqref{eq:limsupPhi} that 
\[
\limsup_{j\ra\infty} \Phi[G^{(j)},U] \leq \lim_{k\ra\infty} a_{j_k} + \ve.
\]
Since $\ve>0$ was arbitrary, this means that
\[
\limsup_{j\ra\infty} \Phi[G^{(j)},U] \leq \mathcal{F}_p[\mathcal{G}_p (\mu_1), U(\,\cdot\,,0)] - \frac{1}{2} \log( (2\pi e)^p \det G_p) + \alpha \Tr[\mc{G}_p(\mu_1) - G_p].
\]
If we had $\mathcal{G}_N(\mu) = G$, i.e., $\mc{G}_p(\mu_1)=G_p$, then we would be done. We have $\mc{G}_p(\mu_1) \preceq G_p$, so it will suffice to show that $\Tr[\mc{G}_p(\mu_1) - G_p]=0$. Suppose for contradiction that $\Tr[\mc{G}_p(\mu_1) - G_1] < 0$. But then, by taking $\alpha$ arbitrarily large we see that $\limsup_{j\ra\infty} \Phi[G^{(j)},U] = -\infty$, which is impossible because we already have a lower bound on $\liminf_{j\ra\infty} \Phi[G^{(j)},U]$. Therefore $\mc{G}_p(\mu_1) = G_p$, as desired, and we have
\[
\limsup_{j\ra\infty} \Phi[G^{(j)},U] \leq \Phi_p[G_p, U(\,\cdot\,,0)],
\]
which completes the proof.

Notice the strong growth assumption was only used 
in this part of the proof (i.e., the proof of the upper bound). In particular, it 
was only used to ensure that the measure $\mu^{(j)}$ of maximum entropy
relative to $\nu_U$ (as in Remark \ref{rem:Fdomain}) subject to the
moment constraint $\mathcal{G}(\mu^{(j)}) = G^{(j)}$ cannot weakly
converge to a measure $\mu$ with $\mathcal{G}(\mu) \neq G = \lim_{j\ra
\infty} G^{(j)}$.

\subsection{Dual perspective on continuous extension}
\label{subsec:alternateView}
We now outline how Theorem \ref{thm:contUpTo} can be
reinterpreted via the transformation rule. 
This perspective provides another way of understanding Theorem
\ref{thm:contUpTo}
and allows us to present a counterexample that illustrates the necessity of the
strong growth condition of Definition \ref{def:Ugrowth2}.

Suppose that $T_j$ are linear transformations such that $T_j \ra P$, where $P = I_p \oplus 0_{N-p}$ is the orthogonal projection onto $\mathrm{span}\{e_1^{(n)},\ldots,e_p^{(n)}\}$. Let $G \in \pd$ with upper-left block given by $G_{p}$. Then, using the transformation rule, Theorem \ref{thm:contUpTo}, and the projection rule, we obtain
\[
\Phi_N [G, U\circ T_j] = \Phi_N [T_j G T_j^*, U] \ra \Phi_p[G_{p}, U(\,\cdot\,,0)] = \Phi_N [G, U \circ P].
\]

This manipulation suggests that Theorem \ref{thm:contUpTo} is equivalent
to the pointwise convergence
\begin{equation}
\label{eq:alternateConv}
\Phi_N [\,\cdot\,, U\circ T_j] \ra \Phi_N [\,\cdot\,, U \circ P]
\end{equation}
for all $T_j \ra P$. 
To see the equivalence, consider an arbitrary sequence $G^{(j)} \in \pd$ 
converging, as before, to the block-diagonal matrix $G = G_p \oplus 0_{N-p} \in \psd$, 
where $G_p \in \mathcal{S}^p_{++}$. Then we want to show, using Eq.~\eqref{eq:alternateConv}, 
that $\Phi_N [ G^{(j)}, U] 
\ra \Phi_p[G_{p}, U(\,\cdot\,,0)]$.

To this end, let $T_j = [G^{(j)}]^{1/2} [G_p \oplus I_{N-p}]^{-1/2}$, so 
$G^{(j)} = T_j (G_p \oplus I_{N-p}) T_j^*$, and $T_j \ra P$. 
Then \eqref{eq:alternateConv} implies that 
$\Phi_N [G_p \oplus I_{N-p}, U\circ T_j] \ra \Phi_N [G_p \oplus I_{N-p}, U \circ P]$, 
and combining with the transformation and projection rules yields
Theorem~\ref{thm:contUpTo}.

Note that \eqref{eq:alternateConv} is equivalent to the pointwise
convergence of concave functions $\mathcal{F}_N[\,\cdot\,,U\circ T_j]
\ra \mathcal{F}_N[\,\cdot\,,U\circ P]$ as $T_j \ra P$. Since the domains
of these concave functions are open (namely, $\pd$), by Theorem \ref{thm:pointwiseUniform} 
this is actually equivalent to uniform convergence on all compact subsets of 
$\pd$. 
Furthermore, since
$\mathcal{F}_N[\,\cdot\,,U\circ T_j]$ and $\mathcal{F}_N[\,\cdot\,,U\circ P]$ 
are both uniformly $-\infty$ on $\symm \backslash \pd$, this is in turn 
equivalent to uniform convergence on all compact subsets of $\symm$ 
that do not contain a boundary point of $\pd$, which by Theorem \ref{epiConvThm} 
is equivalent to the hypo-convergence (see Definition \ref{def:epiConv}) 
$\mathcal{F}_N[\,\cdot\,,U\circ T_j]
\overset{\mr{h}}{\ra} \mathcal{F}_N[\,\cdot\,,U\circ P]$. (Note that the role of 
epi-convergence for convex functions is assumed by hypo-convergence 
for concave functions.)
But then hypo-convergence is equivalent to
hypo-convergence of the concave conjugates (Theorem \ref{conjConvThm}), i.e., 
of 
 $\Omega[\,\cdot\,,U\circ T_j]$ to 
$\Omega[\,\cdot\,,U\circ P]$ as $j \ra
\infty$.

In summary, the continuous extension property is equivalent to the hypo-convergence 
$\Omega[\,\cdot\,,U\circ T_j] \overset{\mr{e}}{\ra} \Omega[\,\cdot\,,U\circ P]$.

\subsection{Counterexample of weak but not strong growth}
\label{subsec:counterexample}

Here we give a counter example to show that the weak growth condition is insufficient for
guaranteeing the continuous extension property. 
By the discussion of section \ref{subsec:alternateView}, we need only find $U$ satisfying the weak 
growth condition for which 
$\Omega[\,\cdot\,,U\circ T_j]$ fails to hypo-converge to
$\Omega[\,\cdot\,,U\circ P]$.

For example, consider $N=2$ and
\[
U(x_1, x_2) = 
\begin{cases}
\vert x_1 \vert^4 & \vert x_1 \vert \leq \vert x_2 \vert^{-1} \\
\vert x_2 \vert^{-4} & \vert x_1 \vert \geq \vert x_2 \vert^{-1}.
\end{cases}.
\]
If $x_2 = 0$, then the first case holds for all $x_1$. This interaction
is nonnegative, and hence satisfies the first part of the weak growth
condition of Definition~\ref{def:Ugrowth} with $C_{U}=0$. To see that $U$ 
satisfies the weak growth condition, we need only show that $\dom\Omega$ is
open. Clearly $\dom\Omega \supset \pd$. Moreover, the restriction of $U$ to 
any line except the $x_1$-axis is bounded, and it follows that in fact 
$\dom\Omega = \pd$, hence $\dom\Omega$ is open, as desired.

Now let
\[
T_j := \left(\begin{array}{cc}
1 & 0\\
0 & j^{-1}
\end{array}\right) \ra P := \left(\begin{array}{cc}
1 & 0\\
0 & 0
\end{array}\right).
\]
Since $\Omega[\,\cdot\,,U\circ P]$ has an open domain, namely,
\[
\dom\left(\Omega[\,\cdot\,,U\circ P]\right)
= \left\{ A = (a_{ij}) \in \mathcal{S}^2 : a_{22} > 0  \right\},
\]
the hypo-convergence of $\Omega[\,\cdot\,,U\circ T_j]$ to $\Omega[\,\cdot\,,U\circ P]$ is equivalent to pointwise convergence (by Theorems \ref{epiConvThm} and \ref{thm:pointwiseUniform}), which is the same as the pointwise convergence $Z[\,\cdot\,,U\circ T_j] \ra Z[\,\cdot\,,U\circ P]$.

Set $A = (a_{ij})$ via $a_{11} = a_{12} = 0$, $a_{22} = 1$, so $A$ is in the domain of $\Omega[\,\cdot\,,U\circ P]$, i.e., $Z[A,U\circ P] < +\infty$. However, 
\[
Z[A,U\circ T_j] = \int e^{- \frac12 \vert x_2 \vert^2 - U(x_1,j^{-1} x_2) }\,\ud x_1\,\ud x_2 = j\cdot \int e^{-j^2 \frac12 \vert x_2 \vert^2 - U(x_1, x_2) }\,\ud x_1\,\ud x_2.
\]
Now the restriction of the last integrand to any line of constant $x_2 \neq 0$ is asymptotically equal to $e^{-j^2 \vert x_2\vert^2 - \vert x_2 \vert^{-4}} > 0$, so the integral along any such line is $+\infty$, and by Fubini's theorem, $Z[A,U\circ T_j] = +\infty$. Thus convergence fails at $A$, and we have a counterexample as claimed.


\section*{Acknowledgments} 

This work was partially supported by the Department of Energy under
grant DE-AC02-05CH11231 (L.L., M.L.), by the Department of Energy under
grant No. DE-SC0017867 and by the Air Force Office of Scientific
Research under award number FA9550-18-1-0095 (L.L.), and by the NSF
Graduate Research Fellowship Program under Grant DGE-1106400 (M.L.).  We
thank Fabien Bruneval, Garnet Chan, Alexandre Chorin, Lek-Heng Lim,
Nicolai Reshetikhin, Chao Yang and Lexing Ying for helpful discussions.


\appendix

\section{Definitions and results from convex analysis}\label{sec:convex}

In this section we review some definitions and results from convex analysis. In this paper 
many results are stated for concave functions, i.e., functions $f$ such that 
$-f$ are convex. The standard results of 
convex analysis can always be applied by considering negations. We state results 
below for convex functions to maintain consistency with the literature. Many results are 
stated in somewhat more generality than is needed 
for the purposes of this paper
(e.g., we do not simply conflate proper and non-proper convex functions). 
This is done to make sure that the reader can refer to the cited references. The
discussion follows developments from Rockafellar \cite{Rock} and 
Rockafellar and Wets \cite{RockWets}.

\subsection{Convex sets and functions}

We begin with the definition of convex sets and functions.

\begin{defn}
A set $C\subset\mathbb{R}^{n}$ is \emph{convex} if $(1-t)x+ty\in C$
for every $x,y\in C$ and all $t\in[0,1]$.
\end{defn}

\begin{defn}
An extended real-valued function $f$ on a convex set $C$, i.e., a function
$f:C \ra [-\infty,\infty] = \R \cup\{ -\infty,+\infty\}$, is \emph{convex} if 
\[
f\left((1-t)x+ty\right)\leq(1-t)f(x)+tf(y)
\]
for all $x,y\in C$
and all $t\in(0,1)$, where we interpret $\infty - \infty = +\infty$ if necessary.
We say that $f$ is 
\emph{strictly convex} on the convex set $C$ if this inequality  
holds strictly whenever $x \neq y$.
\end{defn}

\begin{defn}
The \emph{(effective) domain} of
a convex function $f$ on $S$, denoted $\mathrm{dom}\, f$, is the
set 
$ \dom f=\{x\in S\,:\, f(x)<+\infty\}$.
\end{defn}

The following is an immediate consequence of the preceding definitions:
\begin{lem}
Let $f$ be convex on $S\subset\mathbb{R}^{n}$. Then $\mathrm{dom}\, f$
is convex.
\end{lem}

We note that when $f\in C^{2}(C)$, our definition of convexity coincides
with the definition from multivariate calculus:
\begin{thm}
Let $f\in C^{2}(C)$, where $C\subset\mathbb{R}^{n}$ is open and
convex. Then $f$ is convex on $C$ if and only if the Hessian matrix
$\nabla^{2}f(x)$ is positive semi-definite for all $x\in C$.\end{thm}
\begin{proof}
See Theorem 4.5 of Rockafellar \cite{Rock}.
\end{proof}

Notice that for $f$ convex on a convex set $C\subset\mathbb{R}^{n}$,
we can extend to $\tilde{f}$ defined on $\mathbb{R}^{n}$ by taking $\tilde{f}\vert_{\mathbb{R}^{n}\backslash C}\equiv+\infty$.
It is immediate that $\tilde{f}$
is convex on $\mathbb{R}^{n}$. Thus one loses no generality by
considering only functions that are convex on $\mathbb{R}^{n}$.

The following definitions are helpful for ruling out pathologies:
\begin{defn}
A convex function $f$ is called \emph{proper} if $\mathrm{dom}\, f\neq\emptyset$
and $f(x)>-\infty$ for all $x$.
\end{defn}

We will only ever need to consider proper convex functions.

\begin{defn}
If $f$ is a proper convex function, then $f$ is called \emph{closed} 
if it is also lower semi-continuous. (If $f$ is a non-proper convex function, then 
$f$ is called \emph{closed} if it is either $f\equiv+\infty$
or $f\equiv-\infty$.)
\end{defn}
\begin{rem}
For the fact that this can be taken as the definition, see Theorem 7.1 of \cite{Rock}.
\end{rem}

The convexity of a function guarantees its continuity in a certain
sense:
\begin{thm}
\label{convLipThm}
A convex function $f$ on $\mathbb{R}^{n}$ is continuous relative
to any relatively open convex set in $\mathrm{dom}\, f$. In particular,
$f$ is continuous on $\mathrm{int\, dom}\, f$. In fact, it holds
that a proper convex function $f$ is locally Lipschitz on $\mathrm{int\, dom}\, f$.\end{thm}
\begin{proof}
See Theorems 10.1 and 10.4 of Rockafellar \cite{Rock}.
\end{proof}

\subsection{First-order properties of convex functions}

There is an extension of the notion of differentiability that is fundamental
to the analysis of convex functions.
\begin{defn}
Let $f$ be a convex function on $\mathbb{R}^{n}$.
$y\in\mathbb{R}^{n}$
is called a \emph{subgradient} of $f$ at $x\in \dom f$ if $f(z)\geq f(x)+\left\langle y,z-x\right\rangle $
for all $z\in\mathbb{R}^{n}$. The \emph{subdifferential} of $f$ at
$x\in \dom f$, denoted $\partial f(x)$, is the set of all subgradients of
$f$ at $x$. By convention $\partial f (x) = \emptyset$ for $x \notin \dom f$.
\end{defn}
\begin{thm}
\label{subdiffNonempty}
Let $f$ be a proper convex function. $\partial f(x)$ is a non-empty
bounded set if and only if $x\in\mathrm{int\, dom}\, f$.\end{thm}
\begin{proof}
See Theorem 23.4 of Rockafellar \cite{Rock}.
\end{proof}
 
It is perhaps no surprise that the derivative and the subdifferential of a convex function coincide wherever it is differentiable.
 
\begin{thm}
Let $f$ be a convex function, and let $x\in\mathbb{R}^{n}$ such
that $f(x)$ is finite. If $f$ is differentiable at $x$, then $\nabla f(x)$
is the unique subgradient of $f$ at $x$, where $\nabla$ is the gradient 
defined with respect to the inner product used to define the subgradient. 
Conversely, if $f$ has
a unique subgradient at $x$, then $f$ is differentiable at $x$.\end{thm}
\begin{proof}
See Theorem 25.1 of Rockafellar \cite{Rock}.
\end{proof}

\subsection{The convex conjugate}
A fundamental notion of convex analysis is convex conjugation, which 
extends the older notion of Legendre transformation.
\begin{defn}
Let $f$ be a function $\mathbb{R}^{n}\rightarrow[-\infty,+\infty]$.
Then the \emph{convex conjugate }(or, \emph{Legendre-Fenchel transform})
$f^{*}:\mathbb{R}^{n}\rightarrow[-\infty,+\infty]$ with respect to an inner product 
$\langle \, \cdot\,,\,\cdot\,\rangle$ on $\R^n$ is defined by
\[
f^{*}(y)=\sup_{x}\left\{ \left\langle x,y\right\rangle -f(x)\right\} =-\inf_{x}\left\{ f(x)-\left\langle x,y\right\rangle \right\} .
\]
\end{defn}
\begin{thm}
\label{doubleDual}
Let $f$ be a convex function. Then $f^{*}$ is a closed convex function,
proper if and only if $f$ is proper. Furthermore, if $f$ is closed,
then $f^{**}=f.$\end{thm}
\begin{proof}
See Theorem 12.2 of Rockafellar \cite{Rock}.
\end{proof}

It is an important fact that the subgradients of $f$ and $f^{*}$
are, in a sense, inverse mappings.
\begin{thm}
\label{invSubThm}
If $f$ is a closed proper convex function, then $x\in\partial f^{*}(y)$
if and only if $y\in\partial f(x)$.\end{thm}
\begin{proof}
See Corollary 23.5.1 of Rockafellar \cite{Rock}.
\end{proof}

Roughly speaking, differentiability of a convex function 
corresponds to the strict convexity of its conjugate. Indeed:
\begin{thm}
\label{strictConvDiff}
If $f$ is a closed proper convex function, then the following are equivalent:
\begin{enumerate}
\item $\intdom f$ is nonempty, $f$ is differentiable on $\intdom f$, 
and $\partial f (x) = \emptyset$ for all $x\in \dom f \,\backslash\, \intdom f$.
\item $f^*$ is strictly convex on all convex subsets of 
$\dom \partial f^* := \{ y \,:\, \partial f^* (y) \neq \emptyset \}$.
\end{enumerate}
\end{thm}
\begin{proof}
See Theorem 11.13 of \cite{RockWets}.
\end{proof}

Note that for proper convex $f$, if $\dom f^*$ is open, then $\dom \partial f^* = \dom f^*$
by Theorem \ref{subdiffNonempty}, and 
under the additional assumption that $\dom f$ is open, 
Theorem \ref{strictConvDiff} simplifies to the following:
\begin{thm}
\label{strictConvDiff2}
Let $f$ is a lower semi-continuous, proper convex function, and suppose that $\dom f$ and $\dom f^*$ 
are open. Then the following are equivalent:
\begin{enumerate}
\item $f$ is differentiable on $\dom f$.
\item $f^*$ is strictly convex on $\dom f^*$.

\end{enumerate}
\end{thm}

\subsection{Sequences of convex functions}

Pointwise convergence of convex functions entails a kind of convergence
of their subgradients.
\begin{thm}
\label{subgradConvThm}
Let $f$ be a convex function on $\mathbb{R}^{n}$, and let $C$ be
an open convex set on which $f$ is finite. Let $f_{1},f_{2},\ldots$
be a sequence of convex functions finite on $C$ and converging pointwise
to $f$ on $C$. Let $x\in C$, and let $x_{1},x_{2},\ldots$ be a
sequence of points in $C$ converging to $x$. Then for any $\varepsilon>0$,
there exists $N$ such that 
\[
\partial f_{i}(x_{i})\subset\partial f(x)+B_{\varepsilon}(0)
\]
 for all $i\geq N$.\end{thm}
\begin{proof}
See Theorem 24.5 of Rockafellar \cite{Rock}.
\end{proof}

Besides pointwise convergence, there is in fact another nature of convergence for 
convex functions. This is the notion of \emph{epi-convergence}, which is defined 
(even for non-convex functions) as follows: 
\begin{defn}
\label{def:epiConv}
Let $f_i, f$ be extended-real-valued functions on $\R ^n$. Then 
we say that the sequence $\{f_i\}$ \emph{epi-converges} to $f$, written as 
$f = \mr{e}\lim_{i\ra\infty} f_i$ or $f_i \overset{\mr{e}}{\ra} f$ as $i\ra \infty$, 
if for all $x \in \R^n$, the following two conditions are satisfied:
\begin{equation*}
\begin{split}
\liminf_i f_i (x_i) \geq f(x) \quad \mbox{for every sequence \ $x_i \ra x$}
\\ 
\limsup_i f_i( x_i) \leq f(x)  \quad \mbox{for some sequence \ $x_i \ra x$}.
\end{split}
\end{equation*}
We say that the sequence  $\{f_i\}$ \emph{hypo-converges} to $f$, written as 
$f = \mr{h}\lim_{i\ra\infty} f_i$ or $f_i \overset{\mr{h}}{\ra} f$ as $i\ra \infty$, 
if $\{-f_i\}$ epi-converges to $-f$.
\end{defn} 

The notion of epi-convergence is particularly natural in the theory of convex 
functions; accordingly hypo-convergence is more relevant to concave functions. 
Note also that epi-convergence is neither stronger nor weaker than pointwise 
convergence. However, there is a useful theorem that relates the pointwise 
convergence and epi-convergence
of convex functions.

\begin{thm}
\label{epiConvThm}
Let $f_{i}$ be a sequence of convex functions on $\mathbb{R}^{n}$,
and let $f$ be a  lower semi-continuous convex function on $\mathbb{R}^{n}$
such that $\mathrm{dom}\, f$ has non-empty interior. Then
$f=\mathrm{e}\lim_{i\to \infty}f_{i}$
if and only if the $f_{i}$ converge uniformly to $f$ on every compact set
$C$ that does not contain a boundary point of $\mathrm{dom}\, f$.\end{thm}

\begin{proof}
See Theorem 7.17 of Rockafellar and Wets \cite{RockWets}.
\end{proof}

Under certain mild conditions, the epi-convergence of a sequence of
convex functions is equivalent to the epi-convergence of the corresponding
sequence of conjugate functions. Indeed, the following theorem is a natural 
motivation for considering epi-convergence as opposed to pointwise convergence.
\begin{thm}
\label{conjConvThm}
Let $f_{i}$ and $f$ be lower semi-continuous, proper convex functions
on $\mathbb{R}^{n}$. Then the $f_{i}$ epi-converge to $f$ if and
only if the $f_{i}^{*}$ epi-converge to $f^{*}$.\end{thm}
\begin{proof}
See Theorem 11.34 of Rockafellar and Wets \cite{RockWets}.
\end{proof}

Finally, under certain circumstances one can upgrade mere pointwise 
convergence of convex functions to uniform convergence on compact subsets:
\begin{thm}
\label{thm:pointwiseUniform}
Let $f_{i}$ and $f$ be finite convex functions on an open convex set $O \subset \R^n$, 
and suppose that $f_i \ra f$ pointwise on $O$. Then $f_i$ converges uniformly 
to $f$ on every compact subset of $O$.
\end{thm}
\begin{proof}
See Corollary 7.18 of Rockafellar and Wets \cite{RockWets}.
\end{proof}

\section{Classical results on weak convergence of probability measures}\label{sec:weakConv}
For completeness we recall here several classical results on the weak convergence 
of measures. For reference, see, e.g., Billingsley \cite{Billingsley}.

Let $S$ be a metric space, and let $\mc{P}(S)$ denote the set of 
probability measures on $S$ (equipped with the Borel $\sigma$-algebra). 
We say that a sequence $\mu_k \in \mc{P}(S)$ converges 
weakly to $\mu \in \mc{P}(S)$, denoted $\mu_k \Rightarrow \mu$, if 
$\int f \, \ud \mu_k \ra \int f\,\ud \mu$ as $k\ra \infty$ for all bounded, continuous 
functions $f : S \ra \R$. A number of equivalent characterizations of 
weak convergence are given in the following result, often known as 
the Portmanteau theorem:

\begin{thm}[Portmanteau]
\label{thm:portmanteau}
Let $S$ be a metric space, and let $\mu_k, \mu \in \mc{P}(S)$.
The following are all equivalent conditions for the weak convergence 
$\mu_k \Rightarrow \mu$: 
\begin{enumerate}
\item $\lim_{k\ra \infty} \int f \, \ud \mu_k = \int f\,\ud \mu$ for all bounded, continuous 
functions $f : S \ra \R$.
\item $\liminf_{k\ra \infty} \int f \, \ud \mu_k \geq \int f\,\ud \mu$ for all lower semi-continuous
functions $f : S \ra \R$ bounded from below.
\item $\liminf_{k\ra \infty} \mu_k(U) \geq \mu( U)$ for all open sets $U \subset S$.
\end{enumerate}
\end{thm}
\begin{rem}
There are several other equivalent conditions often included in the 
statement of this result.
\end{rem}

A condition for extracting a weakly convergent subsequence, as guaranteed 
by Prokhorov's theorem below, is given by the following notion of tightness:
\begin{defn}
\label{def:tightness}
Let $S$ be a metric space equipped with the Borel $\sigma$-algebra. A set 
$\mc{C}$ of measures on $S$ is called $\emph{tight}$ if for any $\ve > 0$, there exists a 
compact subset $K \subset S$ such that $\mu(K) > 1-\ve$ for all $\mu \in \mc{C}$. 
A sequence of measures is called tight if the set of terms in the sequence is tight.
\end{defn}

\begin{thm}[Prokhorov]
\label{thm:prokhorov}
Let $S$ be a metric space equipped 
with the Borel $\sigma$-algebra. Then 
any tight sequence in $\mc{P}(S)$ admits a weakly convergent subsequence.
\end{thm}

\section{Proof of Lemmas}\label{sec:appendix}

\subsection{Lemma~\ref{lem:entropyMoment}}

\begin{proof}
Suppose $\mu \ll \lambda$ is in $\mathcal{M}_2$ and write $\ud\mu =
\rho\,\ud x$ where $\rho$ is the probability density. Since $\mu \ll
\lambda$, $\mathrm{Cov}(\mu)$ must be positive definite. Let $\mu_G$ be
the Gaussian measure with the same mean and covariance as $\mu$, and let
$\rho_G$ be the corresponding probability density.
Then one can compute that 
\[
\int \rho \log \rho_G \,\ud x = -\frac{1}{2} \log\left( (2\pi e)^N \det \mathrm{Cov}(\mu) \right)
\]
(and in particular this integral is absolutely convergent). Now 
\[
\rho \log \rho = \rho \log \rho_G  + \rho \log \frac{\rho}{\rho_G}.
\]
The first term on the right-hand side of this equation is absolutely
integrable, and the integral of the second term exists (in particular,
the integral of the negative part of the second term is finite, and the value of the full
integral is in fact $-H_{\mu_G}(\mu)$). Therefore the integral $\int
\rho \log \rho \,\ud x \in (-\infty,\infty]$ exists.  Moreover
\[
H(\mu) = -\int \rho \log \rho \,\ud x = \frac{1}{2} \log\left( (2\pi
e)^N \det \mathrm{Cov}(\mu) \right) + H_{\mu_G}(\mu) \leq \frac12 \log\left( (2\pi e)^N \det \mathrm{Cov}(\mu) \right),
\]
with equality if and only if $\mu_G = \mu$.

To prove the second inequality in the statement of the lemma, define
$\overline{\mu}:=\int x\, \ud \mu$ to be the mean of $\mu$. Then $\mathrm{Cov}(\mu) = \mathcal{G}(\mu) - \overline{\mu}\,\overline{\mu}^T$, so in particular $\det \mathrm{Cov}(\mu) \leq \det \mathcal{G}(\mu)$, with equality if and only if $\overline{\mu} = 0$.
\end{proof}

\subsection{Lemma~\ref{lem:entropyUSC}}

\begin{proof}
Without loss of generality we can assume that $\mu_j = \rho_j \,\ud x$ for all $j$.

First, by the Portmanteau theorem for weak convergence of measures 
(Theorem \ref{thm:portmanteau}) we have, for  any $z\in\Rn$, that 
\begin{eqnarray*}
  z^{T}\mathcal{G}(\mu) z = \int(z^T x)^2 \,\ud\mu & \leq & \liminf_{j\ra \infty} \int(z^T x)^2 \,\ud\mu^{(j)} \\
& = & \liminf_{j\ra \infty} \int z^T xx^T z \,\ud\mu^{(j)}
 =  \liminf_{j\ra \infty} z^T \mathcal{G}(\mu_j) z \leq C \Vert z\Vert^2.
\end{eqnarray*}
It follows that $\mu \in \mathcal{M}_2$ (and moreover $\mathcal{G}(\mu) \preceq C\cdot I_n$).

Our goal is to put ourselves in a position to use the upper
semi-continuity (note our sign convention) of the \emph{relative}
entropy with respect to the topology of weak convergence (see Fact \ref{fact:relEnt}). 
Let $\beta > 0$, and let $Z_\beta = \int e^{-\beta \Vert x\Vert^2}\,\ud
x$. Let $\gamma_\beta$ be the Gaussian measure with density proportional to $e^{-\beta \Vert x\Vert^2}$. Then
\begin{eqnarray*}
H(\mu_j) & = & -\int \rho_j \log \rho_j \,\ud x \\
& = &  \log(Z_\beta) - \int \rho_j(x) \log \frac{\rho_j(x)}{\frac{1}{Z_\beta} e^{-\beta \Vert x\Vert^2 }} \,\ud x + \beta \int \rho_j(x) \Vert x\Vert ^2 \,\ud x \\
& = & \log(Z_\beta) + H_{\gamma_\beta} (\mu_j) + \beta \Tr[\mathcal{G}(\mu_j)].
\end{eqnarray*}

Then by the upper semi-continuity of the \emph{relative} entropy with respect to the topology of weak convergence, we have
\[
\limsup_{j\ra \infty} H(\mu_j) \leq \log(Z_\beta) + H_{\gamma_\beta} (\mu) + \beta C N = H(\mu) + \beta \left(C N - \Tr[\mathcal{G}(\mu)] \right).
\]
Since this inequality holds for any $\beta>0$, the lemma follows.
\end{proof}

\subsection{Fact \ref{fact:productEntropy}}
\begin{proof}
We can assume that $\mu$ is absolutely continuous with respect to the Lebesgue measure, i.e., has a density $\rho$ (otherwise $H(\mu) = -\infty$ and the inequality is trivial). It follows that $\mu_i := \pi_i \# \mu$ are absolutely continuous with respect to the Lebesgue measure, i.e., have densities $\rho_i$, for $i=1,2$. Let $x=(x_1,x_2)$ denote the splitting of 
$x\in \R^N$ according to the product structure $\Rn = \R^p \times \R^{N-p}$. Then using the fact that $\mu_1 \times \mu_2$ has density $\rho_1(x_1) \rho_2(x_2)$, one directly computes that
\begin{equation*}
\begin{split}
& H(\mu_1) + H(\mu_2) + H_{\mu_1 \times \mu_2} (\mu) \\
& \ \  =
\int \rho_1(x_1) \log \rho_1(x_1) \ud x_1 + 
\int \rho_2(x_2) \log \rho_2(x_2) \ud x_2 + \int \rho(x) \log \frac{\rho(x)}{\rho_1(x_1) \rho_2(x_2)} \ud x \\ 
& \ \  =
\int \rho(x) \log \rho_1(x_1) \ud x + 
\int \rho(x) \log \rho_2(x_2) \ud x + \int \rho(x) \log \frac{\rho(x)}{\rho_1(x_1) \rho_2(x_2)} \ud x \\ 
& \ \ = \int \rho(x) \log \rho(x) \ud x\\
& \ \ = H(\mu).
\end{split}
\end{equation*}
But by Fact \ref{fact:relEnt}, the relative entropy term 
is non-negative.
\end{proof}

\subsection{Lemma~\ref{lem:OmegaConcave}}

\begin{proof}
  Upper semi-continuity follows directly from Fatou's
  lemma. $\Omega$ is proper because its domain is nonempty and evidently $\Omega$ does not attain the value $+\infty$.

Now let $\theta\in[0,1]$ and $A_{1},A_{2}\in \dom \Omega$. Then

\begin{eqnarray*}
\Omega[\theta A_{1}+(1-\theta)A_{2}] & = & -\log\int_{\Rn}\left(e^{-\frac{1}{2}x^{T}A_{1}x-U(x)}\right)^{\theta}\left(e^{-\frac{1}{2}x^{T}A_{2}x-U(x)}\right)^{1-\theta}\,\ud x\\
 & \geq & -\log\left[\left(\int_{\Rn}e^{-\frac{1}{2}x^{T}A_{1}x-U(x)}\,\ud x\right)^{\theta}\left(\int_{\Rn}e^{-\frac{1}{2}x^{T}A_{2}x-U(x)}\,\ud x\right)^{1-\theta}\right]\\
 & = & \theta\Omega[A_{1}]+(1-\theta)\Omega[A_{2}],
\end{eqnarray*}
 where we have used H\"older's  inequality in the second step. This establishes concavity. Strict
concavity on $\dom\Omega$ follows from the following fact: H\"older's
inequality holds with equality in this scenario if and only
if $e^{-\frac{1}{2}x^{T}A_{1}x-U(x)}=e^{-\frac{1}{2}x^{T}A_{2}x-U(x)}$
for all $x$, i.e., if and only if $A_{1}=A_{2}$.

Lastly observe that since $\dom\Omega$ is an open set, for any $A\in \dom \Omega$, 
\[
\int_\Rn e^{\delta x^2} e^{-\frac{1}{2}x^{T}Ax-U(x)}\,\ud x < +\infty 
\]
for some $\delta > 0$. Now, for any polynomial $P$, there exists a constant $C$ such that for all $A'$ in a sufficiently small neighborhood of $A$, 
\[
P(x) e^{-\frac{1}{2}x^{T}A'x-U(x)} \leq C e^{\delta x^2} e^{-\frac{1}{2}x^{T}Ax-U(x)}.
\]
Since derivatives of all orders of the integrand in \eqref{eq:OmegaDef} are of the form \[
P(x) e^{-\frac{1}{2}x^{T}Ax-U(x)},
\]
differentiation under the integral is justified, and the smoothness result follows.
\end{proof}

\subsection{Lemma~\ref{lem:OmegaDual1}}

\begin{proof}
First assume $A \in \dom\Omega$, so $Z[A]<+\infty$. Let
$\mu\in\mathcal{M}_2$ and define $f(x):=\frac{1}{2}x^{T}Ax+U(x)$.  For
any $f$ such that $e^{-f}$ is integrable, define $\nu_f$ to be the
probability measure with density proportional to $e^{-f}$.
Then provided $\mu \ll \lambda$, 
\begin{equation}
  \label{eq:entropy}
  \begin{split}
\int f \,\ud\mu -H(\mu) & =  \Omega[A] -\int \log\left(\frac{1}{Z[A]} e^{-f} \right) \,\ud\mu -H(\mu) \\
& =  \Omega[A] + \int
\log\left(\frac{\ud\mu}{\ud\lambda}\right)-\log\left(\frac{\ud\nu_{f}}{\ud\lambda} \right) \,\ud\mu \\
& =  \Omega[A] + \int \log \frac{\ud\mu}{\ud\nu_f} \,\ud\mu \\
& =  \Omega[A] - H_{\nu_f}(\mu) \geq \Omega[A].
  \end{split}
\end{equation}
Since $\mu \in \mathcal{M}_2$, we have $H(\mu) < +\infty$ as discussed
in Remark \ref{rem:infSet}.
Careful observation reveals that manipulations are valid in the sense of
the extended real numbers even when $\int f\,\ud\mu = +\infty$. Moreover, $\mu \not\ll \lambda$ if
and only if $\mu \not\ll \nu_f$, in which case both sides of
\eqref{eq:entropy} are $+\infty$. Therefore \eqref{eq:entropy} holds for
all $\mu \in \mathcal{M}_2$.

For $A\in\dom\Omega$, \eqref{eq:entropy} establishes the `$\leq$' direction of \eqref{eq:omegaInf1}. For $A\notin \dom\Omega$, $\Omega[A]=-\infty$, so this direction is immediate.

Next suppose that $A \in \dom\Omega$. Since $\dom\Omega$ is open, it follows that $\nu_f \in \mathcal{M}_2$. From
\eqref{eq:entropy} and the inequality $-H_{\nu_f}(\mu)\geq 0$ (which
holds with equality if and only if $\mu = \nu_f$), it follows that
\eqref{eq:omegaInf1} holds. Moreover, that the infimum in
\eqref{eq:omegaInf1} is uniquely attained at $\mu = \nu_f$, i.e., at
$d\mu(x)=\frac{1}{Z[A]}e^{-\frac{1}{2}x^{T}Ax-U(x)}\,\ud x$.
\end{proof}

\subsection{Lemma~\ref{lem:Ffinite}}

\begin{proof}
  By definition $\mathcal{F}[G]=-\infty$ whenever $G\in\mathcal{S}^{N}\backslash\mathcal{S}_{+}^{N}$. Now we show that also 
$\mathcal{F}[G]=-\infty$ for $G$ on the boundary $\partial
\mathcal{S}_{+}^{N}$. This follows from the fact that for such $G$, any
$\mu\in \mathcal{G}^{-1}(G)$ is supported on a subspace of $\Rn$ of
positive codimension, i.e., not absolutely continuous with respect to
the Lebesgue measure, and therefore $H(\mu)=-\infty$. Moreover, since
such $\mu$ is in $\mathcal{M}_2$, we have (via the weak growth
condition) that $\int U\,\ud\mu \in (-\infty,\infty]$, so the expression
within the supremum of \eqref{eq:Fdef} is $-\infty$ for all $\mu\in
\mathcal{G}^{-1}(G)$.

Meanwhile, for $G\in S_{++}^{N}$, one can see that $\mathcal{F}[G]>-\infty$ by considering $\mu$ to be mean-zero with a compactly supported smooth density, linearly transformed to have the appropriate covariance $G$. For such $\mu$, both terms in the supremum are finite.

Moreover, for $G\in S_{++}^{N}$ we also have that $\mathcal{F}[G] < +\infty$. Indeed, for $\mu \in \mathcal{G}^{-1}(G)$, by Lemma \ref{lem:entropyMoment} we have $H(\mu) \leq \frac{1}{2}\log\left[(2\pi e)^n \det G \right]$. Since $\int U\,\ud\mu \geq -C_U (1+\mathrm{Tr}\,G)$, we have a finite upper bound on the expression in the supremum in \eqref{eq:Fdef}, which finishes the proof.
\end{proof}

\subsection{Lemma~\ref{lem:Fconcave}}

\begin{proof}
Let $G_{1},G_{2}\in \pd$, $\theta\in[0,1]$, and $\ve>0$.
Furthermore let $\mu_{1},\mu_{2}\in\mathcal{M}_2$ such that $\mu_{i}\in \mathcal{G}^{-1}(G_{i})$
and $\Psi[\mu_{i}]\geq\mathcal{F}[G_{i}]-\ve/2$. Then, noting that
$\theta\mu_{1}+(1-\theta)\mu_{2}\in \mathcal{G}^{-1}\left(\theta G_{1}+(1-\theta)G_{2}\right)$,
we observe 
\begin{eqnarray*}
\mathcal{F}[\theta G_{1}+(1-\theta)G_{2}] & = & \sup_{\mu\in \mathcal{G}^{-1}\left(\theta G_{1}+(1-\theta)G_{2}\right)}\Psi[\mu]\\
 & \geq & \Psi\left[\theta\mu_{1}+(1-\theta)\mu_{2}\right]\\
 & \geq & \theta\Psi[\mu_{1}]+(1-\theta)\Psi[\mu_{2}]\\
 & \geq & \theta\mathcal{F}[G_{1}]+(1-\theta)\mathcal{F}[G_{2}] - \ve,
\end{eqnarray*}
 where the penultimate step employs convexity of $\Psi$. Since $\ve$
was arbitrary, we have established concavity.

The fact that $\mathcal{F}$ is proper follows from Lemma
\ref{lem:Ffinite}.
Since $\mathcal{F}$ is concave, by Theorem \ref{convLipThm}
it is continuous on $\intdom \mathcal{F}$, which is in fact all of  $\dom \mathcal{F}$ 
by the weak growth assumption.
Thus we only need to check upper semi-continuity at points $G$ outside of $\dom \mathcal{F}$. At $G \notin \overline{\dom\mathcal{F}} = \psd$, upper semi-continuity is trivial because $\mathcal{F} \equiv -\infty$ on a neighborhood of $G$.
Therefore let $G\in \partial \pd$ and suppose that $G_k \in \pd$ such that 
$G_k \ra G$ as $k\ra
\infty$. We need to show that $\limsup_{k\ra\infty} \mathcal{F}[G_k] = -\infty$. Throwing out all $G_k \notin \pd$ from the sequence cannot increase the limit superior, so we can just assume that $G_k \in \pd$ for all $k$. Since $G \in \partial \pd$, we have $\det G = 0$, and therefore $\det G_k \ra 0$. By Lemma \ref{lem:Ffinite} we have
\[
\mathcal{F}[G_k] \leq \frac{1}{2}\log\left[(2\pi e)^n \det G_k \right] + C_U (1+\mathrm{Tr}\,G_k).
\]
Since the right-hand side of this inequality goes to $-\infty$ as $k\ra \infty$, the proof is complete.
\end{proof}

\subsection{Lemma~\ref{lem:Fdiff}}

\begin{proof}
Observe that
(1) $\Omega$ and $\mc{F}$ are upper semi-continuous, proper concave functions (by Lemmas \ref{lem:OmegaConcave} and \ref{lem:Fconcave}), (2) $\mathcal{F} = \Omega^*$ and 
$\Omega = \mc{F}^*$,
and (3) both $\dom \Omega$ and $\dom \mc{F} = \pd$ are open. Then 
the strict concavity and differentiability of $\mc{F}$ on $\dom \mc{F} = \pd$ 
follow directly from Theorem \ref{strictConvDiff2}.

Now we turn to proving $\smooth$-smoothness. Though infinite-order
differentiability is not typically discussed in convex
analysis, it can be obtained from infinite-order differentiability and strict convexity of the convex conjugate via the implicit function theorem. Indeed, define the smooth function $h:\mathcal{S}^n_{++}  \times \dom\Omega \ra \mathcal{S}^n$ by
\[
h(G,A)= \nabla\Omega[A] - G.
\]
Then
$Dh=\left(\ -I_{\mathcal{S}^n}\  \big\vert\  \nabla^2 \Omega \ \right)$, and since $\Omega$ is smooth and strictly concave, the right block is invertible for all $A,G$. Fix some $G'\in \mathcal{S}_{++}^n$, and let $A' = \nabla\mathcal{F}[G']\in \dom\Omega$, so $h(G',A') = 0$.
Then the implicit function theorem gives the existence of a smooth function $\phi $ on a neighborhood $\mathcal{V}\subset \mathcal{S}^n_{++}$ of $G'$ such that $h(G,\phi(G))=0$ for all $G\in\mathcal{V}$. But this means precisely that $\phi = \nabla\mathcal{F}$, hence in particular $\nabla\mathcal{F}$ is smooth at $G'$.
\end{proof}

\subsection{Lemma~\ref{lem:Aeps}}

\begin{proof}

Write
\[
Z[A,\ve U] = \int e^{-\frac12 x^T A x - \ve U(x)} \ud x.
\]
We want to show that as $\ve \ra 0^+$, 
$Z[\,\cdot\,,\ve U]$ epi-converges (see Definition \ref{def:epiConv}) to $Z[\,\cdot\,,\ve U]$. 
If so, then $-\Omega[\,\cdot\,,\ve U]$ epi-converges $-\Omega[\,\cdot\,,0]$, and 
Theorems \ref{conjConvThm} and \ref{epiConvThm}
yield in particular that 
$\mc{F}[\,\cdot\,,\ve U] \ra \mc{F}[\,\cdot\,,0]$ pointwise on $\pd$ as $\ve \ra 0^+$. 
Then by Theorem \ref{subgradConvThm} we have the pointwise convergence of the 
gradients on $\pd$, i.e., 
$A[G,\ve U] \ra A[G,0] =  G^{-1}$ as $\ve \ra 0^+$ 
for $G\in \pd$.

Thus it remains to show that $Z[\,\cdot\,,\ve U]$ epi-converges to $Z[\,\cdot\,,\ve U]$. 
The first of the conditions in Definition \ref{def:epiConv} follows immediately from Fatou's 
lemma, so we need only show that for any $A \in \symm$, there exists a sequence 
$A_\ve \ra A$ such that  
\[
\limsup_{\ve\ra 0^+} Z [A_\ve, \ve U ] \leq  Z_\ve [A, 0 ]
\]
In particular, it suffices to show that 
\begin{equation}
\label{eqn:weakGrowthZIneq}
\limsup_{\ve\ra 0^+} Z [A, \ve U ] \leq  Z_\ve [A, 0].
\end{equation}
For $A \notin \pd$, the righthand side is $+\infty$, so the inequality holds trivially. 

Thus assume $A \in \pd$. 
By the weak growth condition, we can write $U(x) = \widetilde{U}(x) - \lambda - \lambda\Vert x\Vert^2$, 
where $C>0$ and $\widetilde{U} \geq 0$. Then 
\[
Z[A,\ve U] =  \int e^{\ve\lambda} e^{-\frac12 x^T (A - \ve \lambda) x - \ve \widetilde{U}(x)} \ud x 
\leq \int e^{\ve\lambda} e^{-\frac12 x^T (A - \ve \lambda) x} \ud x, 
\]
and evidently the righthand side converges to $Z[A,0]$ by dominated convergence.
\end{proof}

\subsection{Lemma~\ref{lem:epsCts}}

\begin{proof}
Let $G\in\pd$. 
Recall Eq.~\eqref{eqn:weakGrowthZIneq} 
from the proof of Lemma \ref{lem:Aeps}. From this inequality, it follows 
that there exists $\tau > 0$ and an open neighborhood  
$\mathcal{N}$ of $G^{-1}$ in $\pd$
such that $A \in \dom \Omega[\,\cdot\,,\ve U]$ for all 
$(\ve , A) \in (0,\tau) \times \mathcal{N}$.

Now consider $\hat{\ve} > 0$ sufficiently small so that $\hat{\ve} < \tau$ and  
$\hat{A} := A_G(\hat{\ve}) \in \mc{N}$ (possible by Lemma \ref{lem:Aeps}). Define the smooth function $h:(0,\tau) \times \mc{N} \ra \symm$ by
\[
h(\ve,A)= \nabla_A \Omega[A,\ve U] - G.
\]
Then
$Dh(\ve,A) =\left(\ * \  \big\vert\  \nabla^2_A \Omega[A,\ve U] \ \right)$, and since $\Omega[\,\cdot\,,\ve U]$ is smooth and strictly concave, the right block is invertible for all $\ve, A$. Moreover, we have $h(\hat{\ve},\hat{A}) = 0$ by construction.
Then the implicit function theorem gives the existence of a smooth
function $\phi $ on a neighborhood $I$ of $\hat{\ve}$ such that
$h(\ve,\phi(\ve))=0$ for all $\ve \in I$. But this means precisely that
$\phi = A_G$. The implicit function theorem then also says that 
\begin{equation}
\label{eq:iftRecursive}
A_G '(\ve) = -(\nabla^2_A \Omega[A_G(\ve),\ve U])^{-1} \frac{\partial h}{\partial \ve}(\ve, A_G(\ve)) 
\end{equation}
for all $\ve \in I$,
where $A_G'$ denotes the ordinary derivative of the function $A_G$ of a single variable.
In particular Eq.~\eqref{eq:iftRecursive} holds at $\ve = \hat{\ve}$. But since $\hat{\ve}$ was arbitrary 
(beyond being taken sufficiently small), it follows that Eq.~\eqref{eq:iftRecursive} simply holds 
for all $\ve > 0$ sufficiently small.

We want to show that all derivatives of $A_G:(0,\infty) \ra \symm$ extend continuously to $[0,\infty)$. 
Starting with $A_G'$, we can examine these functions by taking further derivatives on the righthand side
of Eq.~\eqref{eq:iftRecursive}. The result will be an expression involving integrals of the form 
\[
\int P(x,U(x))\, e^{-\frac12 x^T A_G(\ve) x - \ve U(x)} \, \ud x,
\]
where $P$ is some polynomial, and it suffices to show that such integrals converge to their 
desired limits 
\[
\int P(x,U(x))\, e^{-\frac12 x^T G^{-1} x} \, \ud x.
\]
The argument is by dominated convergence. First observe that from the 
at-most-exponential growth assumption (Assumption \ref{assumption:atMostExp}), 
it follows that there exist $a,b>0$ such that $ \vert P(x,U(x)) \vert \leq a e^{b \Vert x\Vert}$ 
for all $x$.
 As in the proof of Lemma \ref{lem:Aeps}, 
write 
$U(x) = \widetilde{U}(x) - \lambda - \lambda\Vert x\Vert^2$, 
where $C>0$ and $\widetilde{U} \geq 0$. Then 
\begin{eqnarray*}
\vert P(x,U(x))\, e^{-\frac12 x^T A_G(\ve) x - \ve U(x)} \vert 
& \leq &  \vert P(x,U(x)) \vert \, e^{\ve\lambda} e^{-\frac12 x^T (A_G(\ve) - \ve \lambda) x - \ve \widetilde{U}(x)} \\ 
& \leq &  a e^{b \Vert x \Vert} e^{\ve \lambda} e^{-\frac12 x^T (A_G(\ve) - \ve \lambda) x}.
\end{eqnarray*}
Then for all $\ve > 0$ small enough such that $\ve < 1$ and $A_G(\ve) - \ve \lambda \succ \frac12 G^{-1}$, 
we see that the absolute value of the integrand is bounded uniformly by 
\[
a e^{b \Vert x \Vert} e^{ \lambda} e^{-\frac14 x^T G^{-1} x},
\]
which is integrable. This completes the dominated convergence argument, and 
we conclude that all derivatives of $A_G$ extend continuously to $[0,\infty)$.

Next we aim to use the preceding to show that all derivatives of $\Phi_G$ and $\Sigma_G$ 
also extend continuously to $[0,\infty)$.

To this end, recall the Dyson equation
\[
\Sigma_G = A_G - G^{-1},
\]
which requires that the desired extension property of $\Sigma_G$ is equivalent to that of $A_G$, 
which we have already proved.

Now for any $\ve >0$, we have 
\begin{eqnarray*}
\Phi_G (\ve) & = & 2\mc{F}[G,\ve U] - \Tr \log G - N\log(2\pi e) \\
& = &  \Tr[A_G (\ve) G] -2\Omega[A_G (\ve),\ve U] - \Tr \log G - N\log(2\pi e)
\end{eqnarray*}
by Legendre duality, 
from which it follows from our extension property for $A_G$, together 
with the arguments used to establish it, that all derivatives of $\Phi_G$ extend continuously to $[0,\infty)$.
\end{proof}

\subsection{Lemma~\ref{lem:Gsim}}\label{app:Gsim}

\begin{proof}
Based on Eqs.~\eqref{eqn:boldProp} and 
\eqref{eqn:boldProp2}, we want to show that $G[A^{(M)}(\ve), U_\ve^{(M)}] \sim G[A^{(M)}(\ve) , \ve U]$. 
As a first step, we aim to show that $Z[A^{(M)}(\ve), U_\ve^{(M)}] \sim Z[A^{(M)}(\ve) , \ve U]$. 
Indeed, we can write
\begin{eqnarray}
& &Z[A^{(M)}(\ve) , \ve U] - Z[A^{(M)}(\ve), U_\ve^{(M)}] \nonumber
\\
& &  \quad \quad \quad =\  \  \int e^{-\frac12 x^T A^{(M)}(\ve) x - \ve U(x)} \left(1 - e^{- \frac{1}{2} x^T \left[\Sigma_G(\ve) - \Sigma_G^{(\leq M)}(\ve) \right] x}  \right) \ud x. \label{eqn:ZSigmaInt}
\end{eqnarray}
We can choose $C$ such that 
\[
-C \ve^{M+1} \preceq \Sigma_G(\ve) - \Sigma_G^{(\leq M)}(\ve) \preceq C \ve^{M+1}
\]
for all $\ve>0$ sufficiently small.

Now let $R(\ve) = \ve^{-p/2}$ for $p\in (0,1)$. We split the integral in 
\eqref{eqn:ZSigmaInt} into a part over $B_{R(\ve)}(0)$ and another 
part over the complement. The integrand is dominated by $e^{- \delta x^T x}$ 
for some $\delta$ uniform in $\ve$, the integral of which over the complement 
of $B_{R(\ve)}(0)$ decays super-algebraically as $\ve \ra 0$, so we can neglect this 
contribution.

Meanwhile, for $x\in B_{R(\ve)}(0)$, we have 
\[
\left\vert x^T \left[\Sigma_G(\ve) - \Sigma_G^{(\leq M)}(\ve) \right] x \right\vert 
\leq C \ve^{M+1-p}, 
\]
hence there exists $C'$ such that 
\[
\left\vert 1 - e^{- \frac{1}{2} x^T \left[\Sigma_G(\ve) - \Sigma_G^{(\leq M)}(\ve) \right] x}  \right\vert \leq C' \ve^{M+1-p}
\]
for all $x\in B_{R(\ve)}(0)$. Combining with \eqref{eqn:ZSigmaInt} and 
dominated convergence, we have established $Z[A^{(M)}(\ve), U_\ve^{(M)}] \sim Z[A^{(M)}(\ve) , \ve U]$.

This result, together, together with analogous arguments applied to integrals of the form 
\[
\int x_i x_j \,e^{-\frac12 x^T A^{(M)}(\ve) x - \ve U(x)} \left(1 - e^{- \frac{1}{2} x^T \left[\Sigma_G(\ve) - \Sigma_G^{(\leq M)}(\ve) \right] x}  \right) \ud x,
\]
yields $G[A^{(M)}(\ve), U_\ve^{(M)}] \sim G[A^{(M)}(\ve) , \ve U]$.
\end{proof}

\subsection{Lemma~\ref{lem:Ftech}}

\begin{proof}
For convenience, we define
\[
\mathcal{F}_{c} [G] := \sup_{\mu \in \mathcal{G}^{-1}(G)\cap \mathcal{M}_c} \left[H(\mu) - \int U\,\ud\mu  \right].
\]
Evidently $\mathcal{F}_c \leq \mathcal{F}$ and $\mathcal{F}_{c}[G] =
-\infty$ if $G \notin \pd$, so we can restrict attention to $G \in \pd$.

Fix $\ve > 0$. Let $G\in \pd$, so $\mathcal{F}[G]$ is finite, and let $\mu \in \mathcal{M}_2$ such that
\[
H(\mu) - \int U \,\ud\mu \geq \mathcal{F}[G] - \ve/2.
\]
In particular, $H(\mu) \neq -\infty$, so $d\mu = \rho\,\ud x$ for some density $\rho$. Then consider the measure $\mu_R \in \mathcal{M}_c (R)$ given by density $\rho_R := Z_R ^{-1}\cdot \rho\cdot \chi_R$, where $\chi_R$ is the indicator function for $B_R (0)$ and $Z_R = \int_{B_R (0)} \rho \,\ud x$. By monotone convergence, $Z_R \ra 1$. 

Unfortunately we cannot expect $\mc{G}(\mu_R) = G$, but we do have $\mc{G}(\mu_R) \ra G$ (following from 
dominated convergence, together with the finite second moments of $\mu$). We then want to modify $\mu_R$ (keeping its support compact) to construct a nearby measure with the 
correct second moments. 

To this end let $G_R = \tau_R [G - \mc{G}(\mu_R)] + \mc{G}(\mu_R)$, where $\tau_R > 1$ is chosen so that 
$\tau_R \ra +\infty$ and the eigenvalues of $G_R$ remain uniformly bounded away from zero and infinity (possible since $\mc{G}(\mu_R) \ra G$). Note that we have
$G = \tau_R^{-1} G_R + (1 - \tau_R^{-1}) \mc{G}(\mu_R)$.

Now let $\pi \in \mc{M}_2$ be any compactly supported measure with a density and finite entropy, and let $\pi_R = T_R \# \pi$, where 
$T_R$ is a linear transformation chosen so that $\mc{G}(\pi_R) = G_R$. Since the eigenvalues of $G_R$ are uniformly bounded away from zero and infinity, 
the $T_R$ can be chosen to have determinants uniformly bounded away from zero and infinity (which
guarantees that that the $\vert H(\pi_R) \vert$ are uniformly bounded), and $\pi_R$ can be taken to have uniformly bounded support.
Then finally we can define a measure $\nu_R := \tau_R^{-1} \pi_R + (1 - \tau_R^{-1}) \mu_R$, so $\mc{G}(\nu_R) = G$ and $\nu_R$ is compactly supported.

For the proof it suffices to show that 
\begin{equation}
\label{eq:compactLimit}
H(\nu_R) - \int U \,\ud\nu_R \ra H(\mu) - \int U \,\ud\mu
\end{equation}
as $R \ra \infty$.

By the weak growth condition (Definition \ref{def:Ugrowth}), we can choose a constant $C$ such that $\widetilde{U}$ defined by $\widetilde{U}(x) := C(1+\Vert x\Vert ^2) + U(x)$ satisfies $\widetilde{U}(x) \geq \Vert x\Vert^2$. Now
\[
\int (1+\Vert x\Vert ^2) \,\ud\mu_R \ra \int (1+\Vert x\Vert ^2) \,\ud\mu < +\infty
\]
by monotone convergence together with the fact that $Z_R \ra 1$. Furthermore  
\[
\tau_R^{-1} \int (1+\Vert x\Vert ^2) \,\ud\pi_R \ra 0,
\]
so in fact 
\[
\int (1+\Vert x\Vert ^2) \,\ud\nu_R \ra \int (1+\Vert x\Vert ^2) \,\ud\mu < +\infty
\]

Therefore, without loss of generality, we can prove \ref{eq:compactLimit} under the assumption that
$U(x) \geq \Vert x\Vert^2$. But then $\int U \,\ud\mu_R \ra \int U\,\ud\mu$ by monotone convergence, and $\tau_R^{-1} \int U \,\ud\pi_R \ra 0$ since 
the $\pi_R$ have uniformly bounded support, so in fact $\int U\,\ud\nu_R \ra \int U\,\ud\mu$.

Then we need only show that $H(\nu_R) \ra H(\mu)$. Here one verifies from the construction that $\nu_R$
converges weakly to $\mu$, 
and moreover the second moments of $\nu_R,\mu$ are uniformly bounded, 
so by Lemma \ref{lem:entropyUSC}, we have 
$\limsup_R H(\nu_R) \leq H(\mu)$.

But by the concavity of the entropy, we have $H(\nu_R) \geq \tau_R^{-1} H(\pi_R) + (1-\tau_R^{-1}) H(\mu_R)$. 
Now recall that the $\vert H(\pi_R) \vert$ are uniformly bounded in $R$, so $ \tau_R^{-1} H(\pi_R) \ra 0$. 
Thus the statement $\liminf_R H(\nu_R) \geq H(\mu)$ (and hence also  $H(\nu_R) \ra H(\mu)$) 
will follow if we can establish $H(\mu_R) \ra H(\mu)$.

Now
\[
H(\mu_R) = \log(Z_R) - Z_R^{-1}\int_{B_R (0)} \rho \,\log\rho \,\ud x.
\]
But we know $Z_R \ra 1$, so we need only show that 
\[
\int_{B_R (0)} \rho \log \rho \,\ud x \ra \int \rho \log \rho \,\ud x.
\]
From Lemma \ref{lem:entropyMoment}, the negative part of $\rho \log \rho$ is integrable. But then the fact that $H(\mu) > -\infty$ precisely means that the positive part of $\rho \log \rho$ is integrable, i.e., $\rho\log\rho$ is absolutely integrable. Then the desired fact follows from dominated convergence.
\end{proof}

\bibliographystyle{siam}
\bibliography{lwref}

\begin{thebibliography}{10}

\bibitem{AltlandSimons2010}
{\sc A.~Altland and B.~D. Simons}, {\em Condensed matter field theory},
  Cambridge Univ. Pr., 2010.

\bibitem{AmitMartin-Mayor2005}
{\sc D.~J. Amit and V.~Martin-Mayor}, {\em Field theory, the renormalization
  group, and critical phenomena: graphs to computers}, World Scientific
  Publishing Co Inc, 2005.

\bibitem{Baerends2001}
{\sc E.~J. Baerends}, {\em Exact exchange-correlation treatment of dissociated
  ${H}_{2}$ in density functional theory}, Phys. Rev. Lett., 87 (2001),
  p.~133004.

\bibitem{BaymKadanoff1961}
{\sc G.~Baym and L.~P. Kadanoff}, {\em Conservation laws and correlation
  functions}, Phys. Rev., 124 (1961), p.~287.

\bibitem{BenlagraKimPepin2011}
{\sc A.~Benlagra, K.-S. Kim, and C.~P{\'{e}}pin}, {\em {The Luttinger-Ward
  functional approach in the Eliashberg framework: a systematic derivation of
  scaling for thermodynamics near the quantum critical point}}, J. Phys.
  Condens. Matter, 23 (2011), p.~145601.

\bibitem{Billingsley}
{\sc P.~Billingsley}, {\em Probability and measure}, Wiley, 2012.

\bibitem{BloechlPruschkePotthoff2013}
{\sc P.~E. Bl{\"o}chl, T.~Pruschke, and M.~Potthoff}, {\em Density-matrix
  functionals from {G}reen's functions}, Phys. Rev. B, 88 (2013), p.~205139.

\bibitem{DahlenVanVon2005}
{\sc N.~E. Dahlen, R.~{Van Leeuwen}, and U.~{Von Barth}}, {\em {Variational
  energy functionals of the green function tested on molecules}}, in Int. J.
  Quantum Chem., vol.~101, 2005, pp.~512--519.

\bibitem{Elder2014}
{\sc R.~Elder}, {\em {Comment on ``Non-existence of the Luttinger-Ward
  functional and misleading convergence of skeleton diagrammatic series for
  Hubbard-like models''}}, arXiv:1407.6599,  (2014).

\bibitem{GeorgesKotliarKrauthEtAl1996}
{\sc A.~Georges, G.~Kotliar, W.~Krauth, and M.~J. Rozenberg}, {\em Dynamical
  mean-field theory of strongly correlated fermion systems and the limit of
  infinite dimensions}, Rev. Mod. Phys., 68 (1996), p.~13.

\bibitem{GunnarssonRohringerSchaeferEtAl2017}
{\sc O.~Gunnarsson, G.~Rohringer, T.~Sch{\"{a}}fer, G.~Sangiovanni, and
  A.~Toschi}, {\em {Breakdown of Traditional Many-Body Theories for Correlated
  Electrons}}, Phys. Rev. Lett., 119 (2017), p.~056402.

\bibitem{Ismail-Beigi2010}
{\sc S.~Ismail-Beigi}, {\em {Correlation energy functional within the GW -RPA:
  Exact forms, approximate forms, and challenges}}, Phys. Rev. B, 81 (2010),
  pp.~1--21.

\bibitem{KotliarSavrasovHauleEtAl2006}
{\sc G.~Kotliar, S.~Y. Savrasov, K.~Haule, V.~S. Oudovenko, O.~Parcollet, and
  C.~A. Marianetti}, {\em Electronic structure calculations with dynamical
  mean-field theory}, Rev. Mod. Phys., 78 (2006), p.~865.

\bibitem{KozikFerreroGeorges2015}
{\sc E.~Kozik, M.~Ferrero, and A.~Georges}, {\em {Nonexistence of the
  Luttinger-Ward Functional and Misleading Convergence of Skeleton Diagrammatic
  Series for Hubbard-Like Models}}, Phys. Rev. Lett., 114 (2015), p.~156402.

\bibitem{Levy1979}
{\sc M.~Levy}, {\em Universal variational functionals of electron densities,
  first-order density matrices, and natural spin-orbitals and solution of the
  v-representability problem}, Proc. Natl. Acad. Sci., 76 (1979),
  pp.~6062--6065.

\bibitem{Lieb1983}
{\sc E.~H. Lieb}, {\em Density functional for {Coulomb} systems}, Int J.
  Quantum Chem., 24 (1983), p.~243.

\bibitem{LinLindsey2018}
{\sc L.~Lin and M.~Lindsey}, {\em Variational structure of {Luttinger-Ward}
  formalism and bold diagrammatic expansion for {Euclidean} lattice field
  theory}, Proc. Natl. Acad. Sci., 115 (2018), p.~2282.

\bibitem{LuttingerWard1960}
{\sc J.~M. Luttinger and J.~C. Ward}, {\em {Ground-state energy of a
  many-fermion system. II}}, Phys. Rev., 118 (1960), p.~1417.

\bibitem{MartinReiningCeperley2016}
{\sc R.~M. Martin, L.~Reining, and D.~M. Ceperley}, {\em Interacting
  Electrons}, Cambridge Univ. Pr., 2016.

\bibitem{Mermin1965}
{\sc N.D. Mermin}, {\em {Thermal properties of the inhomogeneous electron
  gas}}, Phys. Rev., 137 (1965), p.~A1441.

\bibitem{NegeleOrland1988}
{\sc J.~W. Negele and H.~Orland}, {\em Quantum many-particle systems},
  Westview, 1988.

\bibitem{RassoulAgha-Sepp2015}
{\sc F.~Rassoul-Agha and T.~Sepp{\"a}l{\"a}inen}, {\em A course on large
  deviations with an introduction to Gibbs measures}, American Mathematical
  Society, 2015.

\bibitem{RentropMedenJakobs2016}
{\sc J.~F. Rentrop, V.~Meden, and S.~G. Jakobs}, {\em {Renormalization group
  flow of the Luttinger-Ward functional: Conserving approximations and
  application to the Anderson impurity model}}, Phys. Rev. B, 93 (2016),
  p.~195160.

\bibitem{Rock}
{\sc R.~T. Rockafellar}, {\em Convex analysis}, Princeton University Press,
  1970.

\bibitem{RockWets}
{\sc R.~T. Rockafellar and R.~J.-B. Wets}, {\em Variational analysis},
  Springer, 2009.

\bibitem{SharmaDewhurstLathiotakisEtAl2008}
{\sc S.~Sharma, J.~K. Dewhurst, N.~N. Lathiotakis, and E.~K.~U. Gross}, {\em
  Reduced density matrix functional for many-electron systems}, Phys. Rev. B,
  78 (2008), p.~201103.

\bibitem{TarantinoRomanielloBergerEtAl2017}
{\sc W.~Tarantino, P.~Romaniello, J.~A. Berger, and L.~Reining}, {\em
  {Self-consistent Dyson equation and self-energy functionals: An analysis and
  illustration on the example of the Hubbard atom}}, Phys. Rev. B, 96 (2017),
  p.~045124.

\bibitem{Zinn-Justin2002}
{\sc J.~Zinn-Justin}, {\em Quantum Field Theory and Critical Phenomena},
  Clarendon Pr., 2002.

\end{thebibliography}

\end{document}